\documentclass[11pt,reqno]{amsart}

\usepackage[a4paper]{geometry}
\usepackage[utf8]{inputenc}

\usepackage{amsmath}
\usepackage{amsfonts}
\usepackage{amssymb}
\usepackage{amsthm}
\usepackage{natbib}
\usepackage{color}
\usepackage{bbm,bm}
\usepackage{graphicx}
\usepackage{subfig,float}
\usepackage{booktabs}
\usepackage{enumerate}

\theoremstyle{plain}
\newtheorem{theorem}{Theorem}[section]
\newtheorem{proposition}[theorem]{Proposition}
\newtheorem{condition}[theorem]{Condition}

\theoremstyle{remark}
\newtheorem{remark}[theorem]{Remark}

\numberwithin{equation}{section}

\newcommand{\RR}{\mathbb{R}}
\newcommand{\vc}[1]{\bm{#1}}
\newcommand{\PP}{\mathbb{P}}
\newcommand{\EE}{\operatorname{\mathbb{E}}}

\newcommand{\cop}{\mathbb{C}_n}
\newcommand{\copb}{\mathbb{C}_n^{\beta}}
\newcommand{\ellb}{\widehat{\ell}_{n,k}^{\beta}}
\newcommand{\ellbst}{\widehat{\ell}_{n,k}^{\beta *}}
\newcommand{\I}{\operatorname{\mathbbm{1}}}
\newcommand{\oh}{\mathrm{o}}
\newcommand{\Oh}{\mathrm{O}}
\newcommand{\diff}{\,\mathrm{d}}
\newcommand{\eps}{\varepsilon}

\renewcommand{\leq}{\leqslant}
\renewcommand{\geq}{\geqslant}
\renewcommand{\le}{\leqslant}
\renewcommand{\ge}{\geqslant}

\begin{document}
\title[Empirical beta stable tail dependence function]{An estimator of the stable tail dependence function based on the empirical beta copula}
\author{Anna Kiriliouk}
\author{Johan Segers}
\author{Laleh Tafakori}
\date{\today}

\begin{abstract}
The replacement of indicator functions by integrated beta kernels in the definition of the empirical stable tail dependence function is shown to produce a smoothed version of the latter estimator with the same asymptotic distribution but superior finite-sample performance. The link of the new estimator with the empirical beta copula enables a simple but effective resampling scheme.

\end{abstract}

\address{Universit\'e catholique de Louvain, Institut de statistique, biostatistique et sciences actuarielles, Voie du Roman Pays 20, B-1348 Louvain-la-Neuve, Belgium}
\email{anna.kiriliouk@uclouvain.be}
\address{Universit\'e catholique de Louvain, Institut de statistique, biostatistique et sciences actuarielles, Voie du Roman Pays 20, B-1348 Louvain-la-Neuve, Belgium}
\email{johan.segers@uclouvain.be}
\address{The University of Melbourne, School of Mathematics and Statistics, Melbourne, VIC 3010, Australia}
\email{laleh.tafakori@unimelb.edu.au}
\maketitle

\section{Introduction}

Let $\bm{X} = (X_1, \ldots, X_d)$ be a random vector with continuous marginal cumulative distribution functions $F_1, \ldots, F_d$. The dependence between the $d$ components at high marginal levels is captured by the stable tail dependence function (stdf)
\begin{equation}
\label{eq:stdf}
  \ell(\bm{x})
  =
  \lim_{t \downarrow 0} 
  t^{-1} \PP\{F_1(X_{1}) > 1-tx_{1}\;\ \text{or}\, \ldots \, \text{or}\;\ F_d(X_{d}) > 1 - tx_{d}\},
\end{equation}
for $\vc{x} \in [0, \infty)^d$ \citep{huang1992, dreeshuang1998}. Existence of the limit in \eqref{eq:stdf} is an assumption. The value of $\ell(\bm{x})$ is proportional to the probability that at least one of the $d$ components of $\vc{X}$ exceeds a high threshold. 
The relative magnitudes of the marginal exceedance probabilities are controlled through the vector $\bm{x}$. The function $\ell$ characterizes dependence in the asymptotic distribution theory of sample maxima and multivariate peaks-over-thresholds \citep{beirlant2006statistics, dehaanferreira2006}. The stdf is on equal footing with other objects in multivariate extreme value theory such as the spectral and exponent measures \citep{dehaan1977}, the Pickands dependence function \citep{pickands1981}, the exponent measure function \citep{coles1991}, and the tail copula \citep{schmidt2006}, all of which can be expressed in terms of one another through appropriate transformations.

Although the analysis of multivariate extremes is often performed within the context of a semiparametric model, underlying many (semi)parametric inference methods is the nonparametric estimator that arises if marginal and joint distribution functions in \eqref{eq:stdf} are replaced by their empirical counterparts. \cite{dreeshuang1998} and later \cite{einmahl2012m} and \cite{fougeres2015nr2} established central limit theorems for the resulting estimator, the empirical stable tail dependence function. The latter authors and \citet{beirlant2016} also introduced bias-corrected versions of the empirical stdf. In two dimensions, the stdf is connected via the inclusion--exclusion formula to the tail copula, of which \citet{schmidt2006} and \citet{bucher2013nr2} studied nonparametric estimators and their properties.

The stdf in \eqref{eq:stdf} may also be viewed in terms of the copula \citep{sklar1959fonctions} of $\bm{X}$,
\begin{equation}
\label{eq:C}
  C(\bm{u}) = \Pr\{ F_1(X_1) \le u_1, \ldots, F_d(X_d) \le u_d \},
  \qquad \bm{u} \in [0, 1]^d;
\end{equation}
recall that we assumed that $F_1, \ldots, F_d$ are continuous. Indeed, we have
\begin{equation}
\label{eq:C2stdf}
  \ell(\bm{x}) = \lim_{t \downarrow 0} t^{-1} \{ 1 - C(1 - t \bm{x}) \}.
\end{equation}
Any estimator for $C$ can thus be converted to an estimator for $\ell$. In fact, the already mentioned empirical stdf arises in exactly this way from the empirical copula \citep{deheuvels1979fonction}, which, up to a minor modification, is the sample analogue of \eqref{eq:C}, with empirical distribution functions replacing their population counterparts.

A drawback of the empirical copula is that it is a piecewise constant function, while the true copula $C$ is continuous. The same issue arises for the empirical versus the true stdf. Smoothing is an obvious remedy. Because the domain of a copula is the $d$-dimensional unit cube, \cite{sancetta2004bernstein} and \cite{janssen2012large} introduced and studied a smoothe based on Bernstein polynomials. For a particular choice of the degree of the polynomials, the resulting empirical Bernstein copula has an interesting interpretation in terms of uniform order statistics. Their distributions being beta, \citet{segers2017} coined this copula estimator the empirical beta copula. Not only did they establish a functional central limit theorem for that estimator under quite general conditions, they also showed that in finite samples, the empirical beta copula is more often than not more accurate than the original empirical copula.

In view of these findings, a natural question is whether the application of the beta smoother to the empirical stdf produces a similar small-sample performance improvement while preserving asymptotic properties. It is the aim of our paper to address this question and investigate further properties of this empirical beta stdf.

We introduce the empirical beta stdf in Section~\ref{sec:estimator} and we prove the weak convergence of the rescaled estimation error in Section~\ref{sec:theory} in a set-up that allows for quite general data generating processes and initial estimators. Further, we investigate the finite-sample performance of the empirical beta stdf in a simulation study reported in Section~\ref{sec:simul}. For all models and settings considered, we find that the empirical beta stdf has a lower integrated mean squared error than the empirical stdf. Finally, we propose in Section~\ref{sec:resampling} a novel resampling procedure that exploits the property that the empirical beta copula is itself a genuine copula and we find in simulations that it compares favorably with the direct multiplier method proposed in \citet{bucher2013nr2}. Some longer mathematical proofs are deferred to Appendix~\ref{sec:proofs}.

\section{Estimators}
\label{sec:estimator}

Let $\vc{X}_i= (X_{i1}, \ldots, X_{id})$, $i\in\{1, \ldots, n\}$, be a sample of $n$ observations of $d$ variables. Assume that within each variable, there are no ties. The rank of $X_{ij}$ among $X_{1j}, \ldots, X_{nj}$ is denoted by $R_{ij,n}=\sum_{t=1}^{n} \I \{ X_{tj} \leq X_{ij}\}$.

The (rank-based) empirical copula and the empirical beta copula are
\begin{align}
\label{eq:empbcop}
  \cop( \vc{u} ) 
  &= 
  \frac{1}{n} \sum_{i=1}^n \prod_{j=1}^d 
  \I \{ R_{ij,n} / n \leq u_j \}
  &\text{and}&&
  \copb( \vc{u} )
  &=
  \frac{1}{n} \sum_{i=1}^n \prod_{j=1}^d 
  F_{n,R_{ij,n}}(u_j),
\end{align}
respectively, where $\vc{u} \in [0, 1]^d$ and where
\begin{equation}
\label{eq:betacdf}
  F_{n,r}(u) 
  = \sum_{s=r}^n \binom{n}{s} u^s (1-u)^s,
  \qquad u \in [0, 1], \ r \in \{1, \ldots, n\},
\end{equation}
is the distribution function of the $\operatorname{Beta}(r,n-r+1)$ distribution, which is the law of $r$-th order statistic of an independent random sample of size $n$ drawn from the uniform distribution on $[0, 1]$. The above empirical copula is a popular variation on the original proposal by \citet{deheuvels1979fonction}, while the empirical beta copula \citep{segers2017} is a special case of the empirical Bernstein copula \citep{sancetta2004bernstein} if the degrees of all Bernstein polynomials are equal to the sample size.

In view of \eqref{eq:C2stdf}, the empirical stable tail dependence function and the (novel) empirical beta stable tail dependence function are defined as
\begin{align}
\label{eq:Cb2ellb}
  \ell_{n,k}( \vc{x} )
  &=
  \tfrac{n}{k} \{ 1 - \cop(1 - \tfrac{k}{n} \vc{x}) \},
  &\text{and}&&
  \ell_{n,k}^\beta( \vc{x} )
  &=
  \tfrac{n}{k} \{ 1 - \copb(1 - \tfrac{k}{n} \vc{x}) \},
\end{align}
respectively. Here $k = k_n \in (0, n]$ is a tuning parameter. Asymptotically, we will need to assume that $k \to \infty$ and $k/n \to 0$. Intuitively, $k/n$ is a `bandwidth' defining a neighbourhood around the upper right corner of the unit cube, while $k = n \cdot (k/n)$ is the order of magnitude of the expected number of data points in such a neighbourhood. In what follows, $\vc{x}$ takes values in a bounded set $[0, M]^d$ for some $M > 0$, and we assume that $k$ and $n$ are such that $\frac{k}{n} x_j \le 1$ for all $j \in \{1, \ldots, d\}$.

The expression of the empirical stable tail dependence function can be developed into
\begin{equation}
\label{eq:ellnk:ranks}
  \ell_{n,k}( \vc{x} )
  =
  \frac{1}{k} \sum_{i=1}^n 
  \I \left\{ 
    R_{i1,n} > n - k x_1\ \text{or $\ldots$ or}\ R_{id,n} > n - k x_d 
  \right\}.
\end{equation}
Replacing $n - k x_j$ by $n + 0.5 - k x_j$ or by $n + 1 - k x_j$ yields variations which are asymptotically equivalent but often more accurate in finite samples. In the bivariate case, replacing `or' by `and' yields the empirical tail copula \citep{schmidt2006}. The empirical stdf originates from the PhD thesis by \citet{huang1992} and, together with the tail copula, it has been studied an applied ever since. Examples include testing for the maximal domain of attraction condition \citep{einmahl2006} and detecting structural breaks in tail dependence for multivariate time series \citep{bucher2015}.

Unlike the empirical copula, the empirical beta copula is itself a copula, i.e., the cumulative distribution function of a random vector whose marginal distributions are uniform on $[0, 1]$. In contrast, the empirical beta stdf is in general not a valid stdf. For instance, it is not necessarily convex and it is not homogeneous; see \citet{ressel2013} for a complete characterization of the class of stdfs. Still, the empirical beta stdf is an improvement upon the ordinary empirical stdf in the sense that it respects the pointwise upper and lower bounds for stdfs and that it has the correct margins.

\begin{proposition}
For $\vc{x} \in [0, n/k]^d$, we have
\[
  \max(x_1, \ldots, x_d) \le \ell_{n,k}^{\beta}(\vc{x}) \le x_1+\cdots+x_d.
\]
In particular, if there exists $j$ such that $x_i = 0$ for all $i \ne j$, then $\ell_{n,k}^{\beta}(\vc{x}) = x_j$.
\end{proposition}

\begin{proof}
Since $\copb$ is itself the cumulative distribution function of a vector of random variables that are uniformly distributed on $[0, 1]$, we have, for $\vc{u} \in [0, 1]^d$,
\begin{align*}
  \copb( \vc{u} ) &\le \min(u_1, \ldots, u_d)
  &\text{and}&&
  1 - \copb( \vc{u} ) &\le (1-u_1) + \cdots + (1-u_d).
\end{align*}
Plugging these inequalities into \eqref{eq:Cb2ellb} with $\vc{u} = 1 - \frac{k}{n} \vc{x}$ yields the proposition.
\end{proof}

We now provide a representation of the empirical beta copula and empirical beta stdf as a mixture over the empirical copula and empirical stdf, respectively, evaluated over a grid of points. Let $\vc{T} = (T_1, \ldots, T_d)$ be a vector of independent binomial random variables such that $T_j \sim \operatorname{Bin}(n, \frac{k}{n}x_j)$ for $j \in \{1, \ldots, d\}$. Recall $F_{n,r}$ in \eqref{eq:betacdf}. With some abuse of notation, we have
\[
  F_{n,r}(1 - \tfrac{k}{n}x)
  =
  \PP[ \operatorname{Bin}(n, 1 - \tfrac{k}{n}x) \ge r ]
  =
  \PP[ \operatorname{Bin}(n, \tfrac{k}{n}x) \le n - r ].
\]
Then by \eqref{eq:empbcop},
\begin{align*}
  \copb(1 - \tfrac{k}{n} \vc{x})
  &=
  \frac{1}{n} \sum_{i=1}^n \prod_{j=1}^d \PP_{\vc{T}}[ T_j \le n - R_{ij}^{(n)} ] \\
  &=
  \EE_{\vc{T}} \biggl[
    \frac{1}{n} \sum_{i=1}^n \prod_{j=1}^d \I \{ T_j \le n - R_{ij}^{(n)} \}
  \biggr] 
  =
  \EE_{\vc{T}}[ \cop( 1 - \vc{T}/n ) ],
\end{align*}
where the notation $\PP_{\vc{T}}$ and $\EE_{\vc{T}}$ means that we compute probabilities and expectations with respect to $\vc{T}$ only. We obtain
\begin{equation}
\label{eq:ellbeta:T}
  \ell_{n,k}^{\beta}( \vc{x} )
  =
  \EE_{\vc{T}}[ \tfrac{n}{k} \{ 1 - \cop( 1 - \vc{T}/n ) \} ]
  =
  \EE_{\vc{T}}[ \ell_{n,k}( \vc{T}/k )].
\end{equation}
Because the notation involving $\vc{T}$ may lead to confusion later on, we let $\nu_{n,k,\vc{x}}$ denote the joint distribution of the random vector $\vc{T}/k$. From \eqref{eq:ellbeta:T}, we obtain the formula
\begin{equation}
\label{eq:ellbeta:nu}
  \ell_{n,k}^{\beta}( \vc{x} )
  =
  \int_{[0, n/k]^d} \ell_{n,k}( \vc{y} ) \diff \nu_{n,k,\vc{x}}(\vc{y}),
  \qquad \vc{x} \in [0, n/k]^d,
\end{equation}
which will be the basis of the asymptotic theory. It also follows that the estimator $\ell_{n,k}^{\beta}$ can be viewed as a smoothed version of $\ell_{n,k}$, with a bandwidth of the order $k^{-1/2}$.

\section{Weak convergence}
\label{sec:theory}

Let $\ell$ be the true stdf, to be estimated. Let $k = k_n \in (0, n]$ be such that $k \to \infty$ and $k/n \to 0$ as $n \to \infty$. Assume we are given a sequence $\hat{\ell}_{n,k}$ of estimators of $\ell$ defined on $[0, n/k]^d$. Here, we make abstraction of the data generating process (independent random sampling, a finite stretch of a stationary time series, \ldots) and of the precise definition of $\hat{\ell}_{n,k}$. The latter could be equal to the empirical stdf in \eqref{eq:Cb2ellb} but also to any variants of it, such as the bias-reduced estimators in \citet{fougeres2015nr2} and \citet{beirlant2016}. We apply the beta smoother $\nu_{n,k,\vc{x}}$ in \eqref{eq:ellbeta:nu} to $\hat{\ell}_{n,k}$, obtaining 
\begin{equation}
\label{eq:hatellbeta}
  \hat{\ell}_{n,k}^\beta( \vc{x} )
  =
  \int_{[0, n/k]^d} \hat{\ell}_{n,k}( \vc{y} ) \diff \nu_{n,k,\vc{x}}(\vc{y}),
  \qquad \vc{x} \in [0, n/k]^d,
\end{equation}
Our aim is to find the asymptotic distribution as $n \to \infty$ of 
\[
  B_{n,k}^\beta = \sqrt{k} ( \hat{\ell}_{n,k}^\beta - \ell ).
\]
Required are smoothness of $\ell$, a growth condition on $\hat{\ell}_{n,k}$, and weak convergence of
\[
  B_{n,k} = \sqrt{k}( \hat{\ell}_{n,k} - \ell ).
\]
We will find that the difference between $B_{n,k}$ and $B_{n,k}^\beta$ is asymptotically negligeable, so that both converge weakly to the same limit.



\begin{condition}
\label{cond:stdf}
For every $j \in \{1, \ldots, d\}$, the function $\ell$ has a continuous first-order partial derivative $\dot{\ell}_j = \partial \ell/ \partial x_j$ on the set $\{ \vc{x} \in [0, \infty)^d : x_j > 0 \}$.
\end{condition}

Condition~\ref{cond:stdf} is satisfied by the members of many parametric families of stdfs. Examples include the logistic model and extensions thereof \citep{tawn1990} and families derived from Gaussian max-stable processes \citep{genton2011, huser2013}. Notable exceptions are stdfs derived from max-linear models \citep{einmahl2012m, gissibl2015, einmahl2016nr2}. See Section~\ref{sec:simul} for specific examples.

\begin{condition}
\label{cond:growth}
We have $\sup_{\vc{y} \in [0, n/k]^d} \lvert \hat{\ell}_{n,k}(\vc{y}) \rvert = \Oh_{\PP}(n/k)$ as $n \to \infty$.
\end{condition}

For $\hat{\ell}_{n,k}$ equal to the empirical stdf $\ell_{n,k}$ in \eqref{eq:Cb2ellb}, Condition~\ref{cond:growth} is trivially satisfied. Condition~\ref{cond:growth} is also satisfied for the bias-corrected estimators in \citet{fougeres2015nr2} and \citet{beirlant2016}: these estimators are based on
\[
  \ell_{n,k;a}( \vc{x} ) = a^{-1} \, \ell_{n,k}(a \vc{x})
\]
for $0 < a \le 1$, and by \eqref{eq:ellnk:ranks}, we have $\lvert \ell_{n,k;a}(\vc{x}) \rvert \le d (x_1 + \cdots + x_d)$ for all $\vc{x} \in [0, \infty)^d$.

For a set $\mathbb{T}$, let 
$\ell^\infty( \mathbb{T} )$ be the Banach space of bounded functions $f : \mathbb{T} \to \RR$, the space being equipped with the supremum norm, $\lVert f \rVert_\infty = \sup \{ \lvert f(\vc{x}) \rvert : \vc{x} \in \mathbb{T} \}$. Weak convergence in $\ell^\infty( \mathbb{T} )$ is denoted by the arrow $\leadsto$ and is to be understood as in \citet[pages~16--28]{van1996weak}.


\begin{condition}
\label{cond:B}
There exists $\delta > 0$ and a stochastic process $B$ on $[0, 1+\delta]^d$ with continuous trajectories such that $B_{n,k} \leadsto B$ as $n \to \infty$ in $\ell^\infty([0, 1+\delta]^d)$.
\end{condition}

For the empirical stdf $\ell_{n,k}$ in \eqref{eq:Cb2ellb}, \citet[Theorem~3.6]{einmahl2012m} showed that $\sqrt{k}(\ell_{n,k} - \ell)$ converges weakly to a Gaussian process with continuous trajectories if the observations $\vc{X}_i$ are sampled independently from a common, continuous distribution function $F$ with stdf $\ell$. All that is required is that $\ell$ is continuously differentiable in the sense of Condition~\ref{cond:stdf} and that the sequence $k = k_n$ grows at a rate that is compatible with the speed at which the limit in \eqref{eq:stdf} is attained. Assuming an additional second-order refinement of \eqref{eq:stdf}, \citet{fougeres2015nr2} and \citet{beirlant2016} proved the weak convergence of the normalized estimation error of their bias-corrected estimators.


\begin{theorem}
\label{thm:weak}
If $k = k_n \in (0, n]$ is such that $\log(n) = \oh(k)$ and $k = \oh(n)$ as $n \to \infty$ and if Conditions~\ref{cond:stdf}, \ref{cond:growth} and~\ref{cond:B} hold, then, in the space $\ell^\infty([0, 1]^d)$, we have
\[
  \sqrt{k} ( \hat{\ell}_{n,k}^{\beta} - \ell )
  =
  \sqrt{k} ( \hat{\ell}_{n,k} - \ell ) + \oh_{\PP}(1)
  \leadsto
  B,
  \qquad n \to \infty.
\]
\end{theorem}

\begin{proof}
Recall the integral representation for $\hat{\ell}_{n,k}^{\beta}$ in \eqref{eq:hatellbeta}. Since $\nu_{n,k,\vc{x}}$ is a probability measure, we have
\begin{align}
\nonumber
  B_{n,k}^{\beta}(\vc{x})
  &=
  \sqrt{k} \left\{ \int_{[0, n/k]^d} \hat{\ell}_{n,k}( \vc{y} ) \diff \nu_{n,k,\vc{x}}(\vc{y}) - \ell(\vc{x}) \right\} \\
\nonumber
  &=
  \int_{[0, n/k]^d} \sqrt{k} \{ \hat{\ell}_{n,k}( \vc{y} ) - \ell(\vc{y}) \} \diff \nu_{n,k,\vc{x}}(\vc{y})
  +
  \sqrt{k} \left\{ \int_{[0, n/k]^d} \ell( \vc{y} ) \diff \nu_{n,k,\vc{x}}(\vc{y}) - \ell(\vc{x}) \right\} \\
\label{eq:Bdecomp}
  &=
  \int_{[0, n/k]^d} B_{n,k}(\vc{y}) \, \diff \nu_{n,k,\vc{x}}(\vc{y})
  +
  \int_{[0, n/k]^d} \sqrt{k} \{ \ell( \vc{y} ) - \ell(\vc{x}) \} \diff \nu_{n,k,\vc{x}}(\vc{y}).
\end{align}
The conclusion of the theorem now follows from Propositions~\ref{prop:stoch} and~\ref{prop:bias} below:
\begin{itemize}
\item 
The first integral in \eqref{eq:Bdecomp} is equal to $B_{n,k}(\vc{x}) + \oh_{\mathbb{P}}(1)$ as $n \to \infty$ by Proposition~\ref{prop:stoch}.
\item 
The second integral in \eqref{eq:Bdecomp} converges to zero uniformly in $\vc{x} \in [0, 1]^d$ by Proposition~\ref{prop:bias} with $f = \ell$. To apply that proposition, we need to verify that the partial derivatives of $\ell$ are uniformly bounded (by one). But this is a consequence of the Lipschitz property 
$\lvert \ell(\vc{y}) - \ell(\vc{x}) \rvert \le \sum_{j=1}^d \lvert y_j - x_j \rvert$ for $\vc{x}, \vc{y} \in [0, \infty)^d$, which is in turn a consequence of \eqref{eq:C2stdf} and the same Lipschitz property for the copula $C$.
\end{itemize}
This concludes the proof of Theorem~\ref{thm:weak}.
\end{proof}

The following two propositions are instrumental in the proof of Theorem~\ref{thm:weak}. Their proofs are given in Appendix~\ref{sec:proofs}. 

\begin{proposition}
\label{prop:stoch}
If $k = k_n \in (0, n]$ is such that $\log(n) = \oh(k)$ and $k = \oh(n)$ as $n \to \infty$ and if Conditions~\ref{cond:growth} and~\ref{cond:B} hold, then
\begin{equation}
\label{eq:BnknuB}
  \sup_{\vc{x} \in [0, 1]^d} 
  \left\lvert
    \int_{[0, n/k]^d}
      B_{n,k}( \vc{y} )
    \diff \nu_{n, k, \vc{x}}( \vc{y} )
    -
    B_{n,k}( \vc{x} )
  \right\rvert
  = 
  \oh_{\PP}(1),
  \qquad n \to \infty.
\end{equation}
\end{proposition}


\begin{proposition}
\label{prop:bias}
If $f : [0, \infty)^d \to \RR$ is continuous and if for each $j \in \{1, \ldots, d\}$, its first-order partial derivative $\dot{f}_j$ exists and is continuous and uniformly bounded on the set $\{ \vc{x} \in [0, \infty)^d : x_j > 0 \}$, then, for $k = k_n \in (0, n]$ such that $k \to \infty$ and $k/n \to 0$ as $n \to \infty$, we have, for every $M > 0$,
\[
  \lim_{n \to \infty}
  \sup_{\vc{x} \in [0, M]^d}
  \left\lvert
    \int_{[0, n/k]^d} \sqrt{k} \{ f( \vc{y} ) - f( \vc{x} ) \} \diff \nu_{n,k,\vc{x}}(\vc{y})
  \right\rvert
  = 0.
\]
\end{proposition}

\begin{remark}
If $\ell$ is not continuously differentiable in the sense of Condition~\ref{cond:stdf}, then \citet{buecher2014nr2} proved that under independent random sampling from a continuous distribution with stdf $\ell$, the sequence of processes $\sqrt{k}(\ell_{n,k}-\ell)$ based upon the empirical stdf $\ell_{n,k}$ still converges weakly in the hypi-topology. The latter topology is weaker than the topology induced by the supremum norm but still stronger than the one induced by the $L^p$ seminorm for $1 \le p < \infty$. An interesting question is whether $\sqrt{k}(\ell_{n,k}^\beta-\ell)$ converges in the hypi-topology too. If true, then functions $\ell$ belonging to max-linear models would be covered as well.
\end{remark}

\begin{remark}
\cite{einmahl2006} proved weak convergence of the (bivariate) empirical stdf with respect to the topology of a weighted supremum norm, providing more information on the estimator in a neighbourhood of the origin. Similarly, \citet{berbucvol17} and \citet{berghaus2017} established weighted weak convergence for the empirical copula and empirical beta copula, respectively. It is an open problem whether such weighted weak convergence also holds for the empirical stdf and empirical beta stdf in arbitrary dimensions.
\end{remark}

\section{Finite-sample performance}
\label{sec:simul}

For independent random sampling, we compared the finite-sample performance of the empirical stdf $\ell_{n,k}$ defined as in \eqref{eq:ellnk:ranks} with $n - kx_j$ replaced by $n + 0.5 - kx_j$ on the one hand and the empirical beta stdf $\ell_{n,k}^{\beta}$ in \eqref{eq:Cb2ellb} defined via the empirical beta copula on the other hand. The simulations were performed with the help of \textsf{R} \citep{Rlanguage}, in particular the packages \textsf{copula} \citep{KojYan10R} and \textsf{SpatialExtremes} \citep{ribatet2017}.


For a given stdf $\ell$, the pseudo-random numbers were sampled independently from the following $d$-variate max-stable distribution function:
\begin{equation*}
  F( \vc{z} ) = \exp \{ - \ell(1/z_1, \ldots, 1/z_d) \}, \qquad \vc{z} \in (0, \infty)^d.
\end{equation*}
The margins of $F$ are unit-Fr\'echet, $F_j(z_j) = \exp(-1/z_j)$ for $z_j \in (0, \infty)$, but actually this choice is immaterial, since both $\ell_{n,k}$ and $\ell_{n,k}^{\beta}$ are functions of the ranks $R_{ij,n}$ only and hence are invariant with respect to increasing transformations of the variables. 

We considered the following max-stable models, in dimensions $d \in \{2, 3, 4\}$:
\begin{itemize}
\item The logistic model in $d = 2$ with $\theta = 0.7$:
\begin{equation}
\label{eq:logistic}
  \ell(x_1, x_2) = (x_1^{1/\theta}+x_2^{1/\theta})^{\theta}.
\end{equation}
\item The max-linear model in $d = 3$ with $2$ factors and parameter vector $\vc{\theta} = (0.3,0.5,0.9)$:
\[
  \ell(x_1, x_2, x_3) 
  = 
  \max \{ \theta_1 x_1, \theta_2 x_2, \theta_3 x_3 \}
  + 
  \max \{ (1-\theta_1) x_1, (1-\theta_2) x_2, (1-\theta_3) x_3 \}.
\]
This function is piece-wise linear and therefore does not satisfy Condition~\ref{cond:stdf}.
\item The law of the Brown--Resnick process on a $2 \times 2$ unit distance grid $\{ \vc{s}_1, \ldots, \vc{s}_4 \} = \{ (0, 0), (1, 0), (0, 1), (1, 1) \}$ with parameters $\rho = \alpha = 1$:
\begin{equation}
\label{eq:BR}
  \ell(x_1, \ldots, x_4) = \EE [ \max \{ x_1 W_1, \ldots, x_4 W_4 \} ]
\end{equation}
where $W_j = \exp \{ \varepsilon(\vc{s}_j) - \gamma(\vc{s}_j) \}$ and where $\{\varepsilon(\vc{s}) : \vc{s} \in \mathbb{R}^2 \}$ is a mean-zero Gaussian random field with $\varepsilon(0, 0) = 0$ almost surely and with semivariogram $\frac{1}{2} \EE[ \{ \varepsilon(\vc{s}) - \varepsilon(\vc{t}) \}^2 ] = \gamma(\vc{s} - \vc{t})$ given by $\gamma(\vc{h}) = \{(h_1^2+h_2^2)^{1/2}/\rho\}^\alpha$ \citep{kabluchko2009}.
\end{itemize}


We measured the performance of the estimators evaluated at $\vc{x} \in [0, 1]^d$ in terms of the integrated squared bias, the integrated variance, and the integrated mean squared error (MSE). For any estimator $\tilde{\ell}_{n,k}$ of $\ell$, these are defined as follows:
\begin{align*}
\textnormal{integrated squared bias:} \qquad & \int_{[0,1]^d} \left\{ \EE [ \tilde{\ell}_{n,k} (\vc{x})] - \ell (\vc{x})  \right\}^2 \diff \vc{x}; \\
\textnormal{integrated variance:} \qquad & \int_{[0,1]^d} \EE \left[ \left\{ \tilde{\ell}_{n,k} (\vc{x}) - \EE [ \tilde{\ell}_{n,k} (\vc{x})] \right\}^2  \right] \diff \vc{x}; \\
\textnormal{integrated MSE:} \qquad & \int_{[0,1]^d} \EE \left[ \left\{ \tilde{\ell}_{n,k} (\vc{x}) - \ell(\vc{x}) \right\}^2  \right] \diff \vc{x}.
\end{align*}
For each of these three quantities, we used the trick based on Fubini's theorem described in \citet[Appendix B]{segers2017} to combine the integral over $[0, 1]^d$ with the expectation over the random sample into a single expectation. As was done there, we computed these expectations using a simple Monte Carlo procedure based upon $20 \, 000$ pseudo-random samples.

\begin{figure}
\centering
\subfloat{\includegraphics[width=0.33\textwidth]{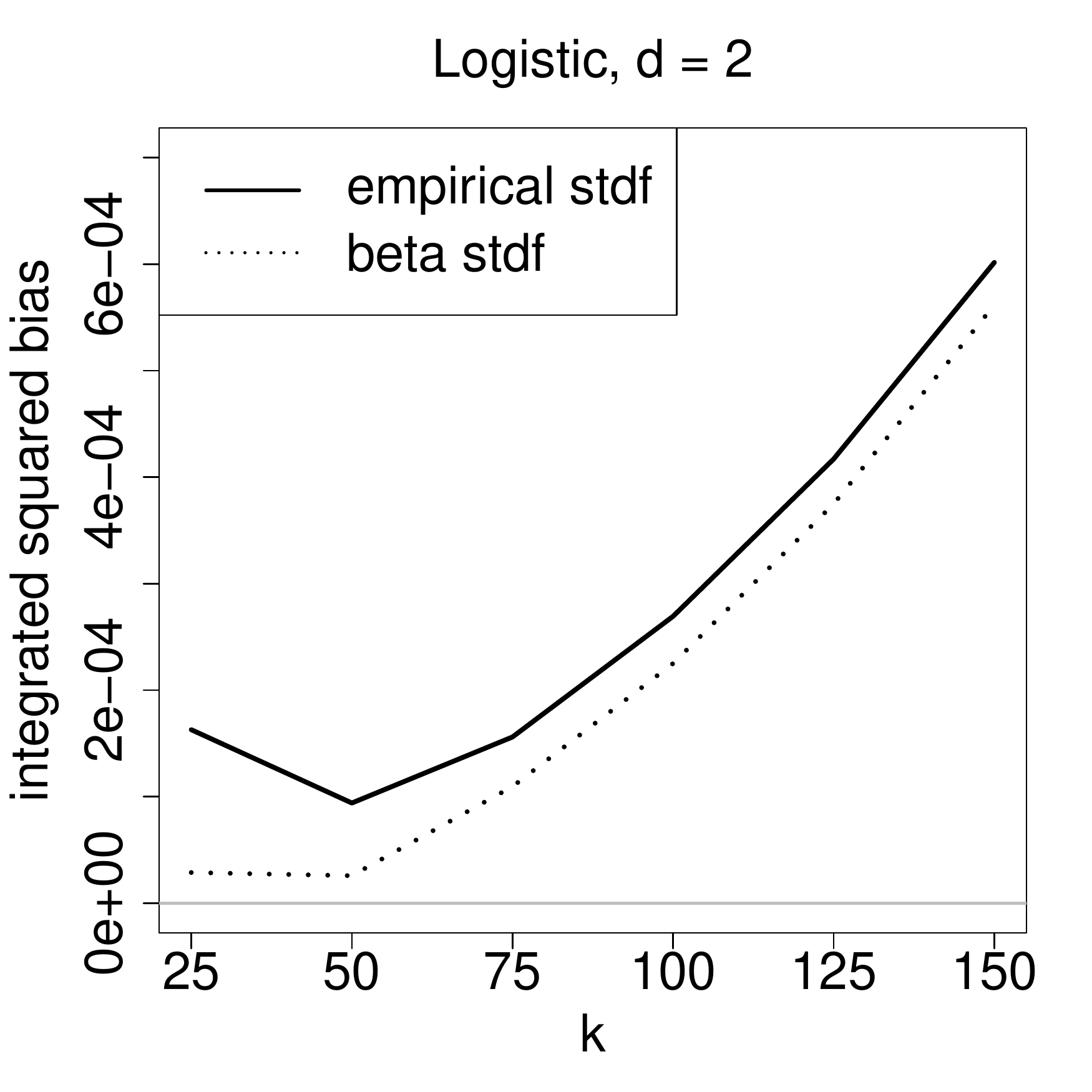}} 
\subfloat{\includegraphics[width=0.33\textwidth]{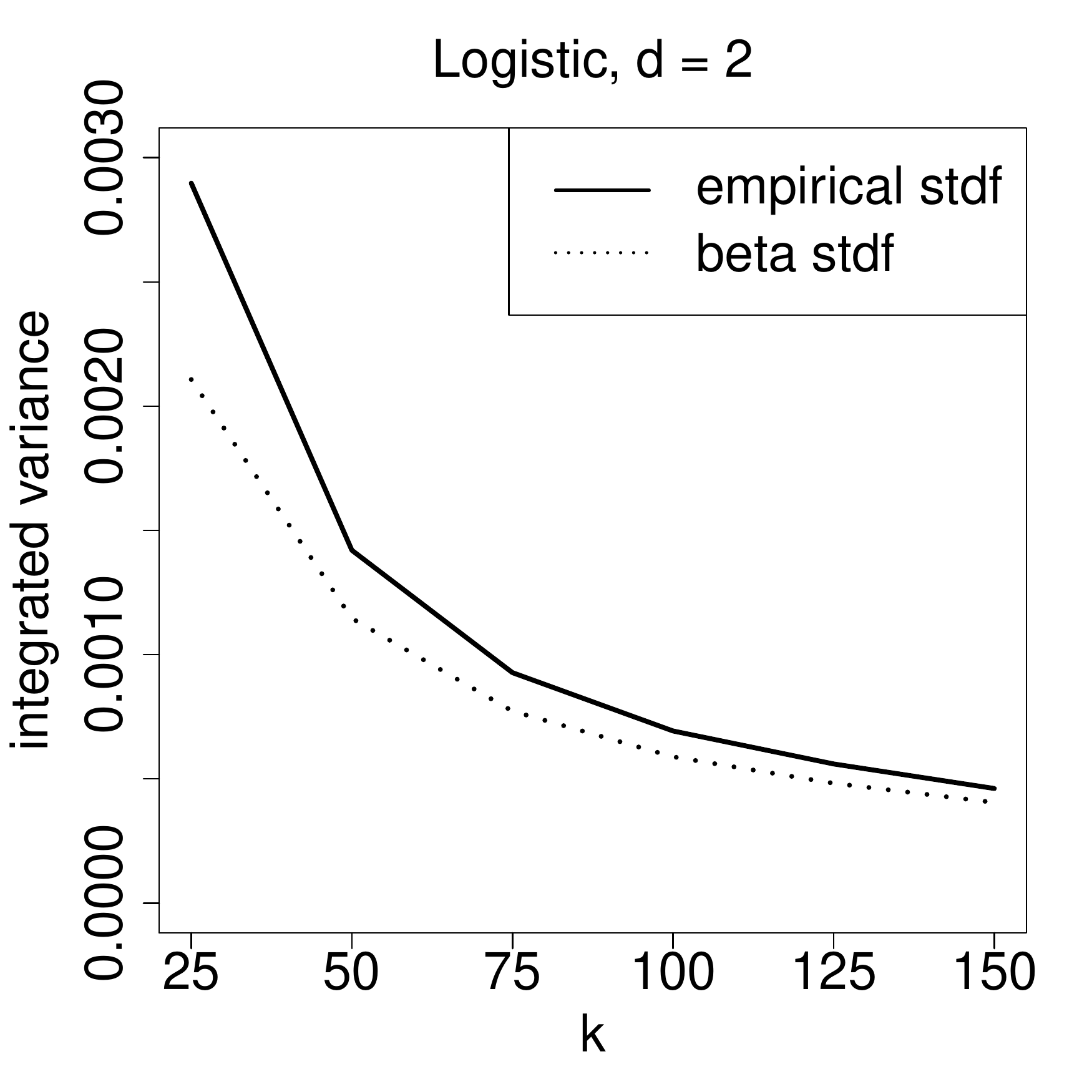}} 
\subfloat{\includegraphics[width=0.33\textwidth]{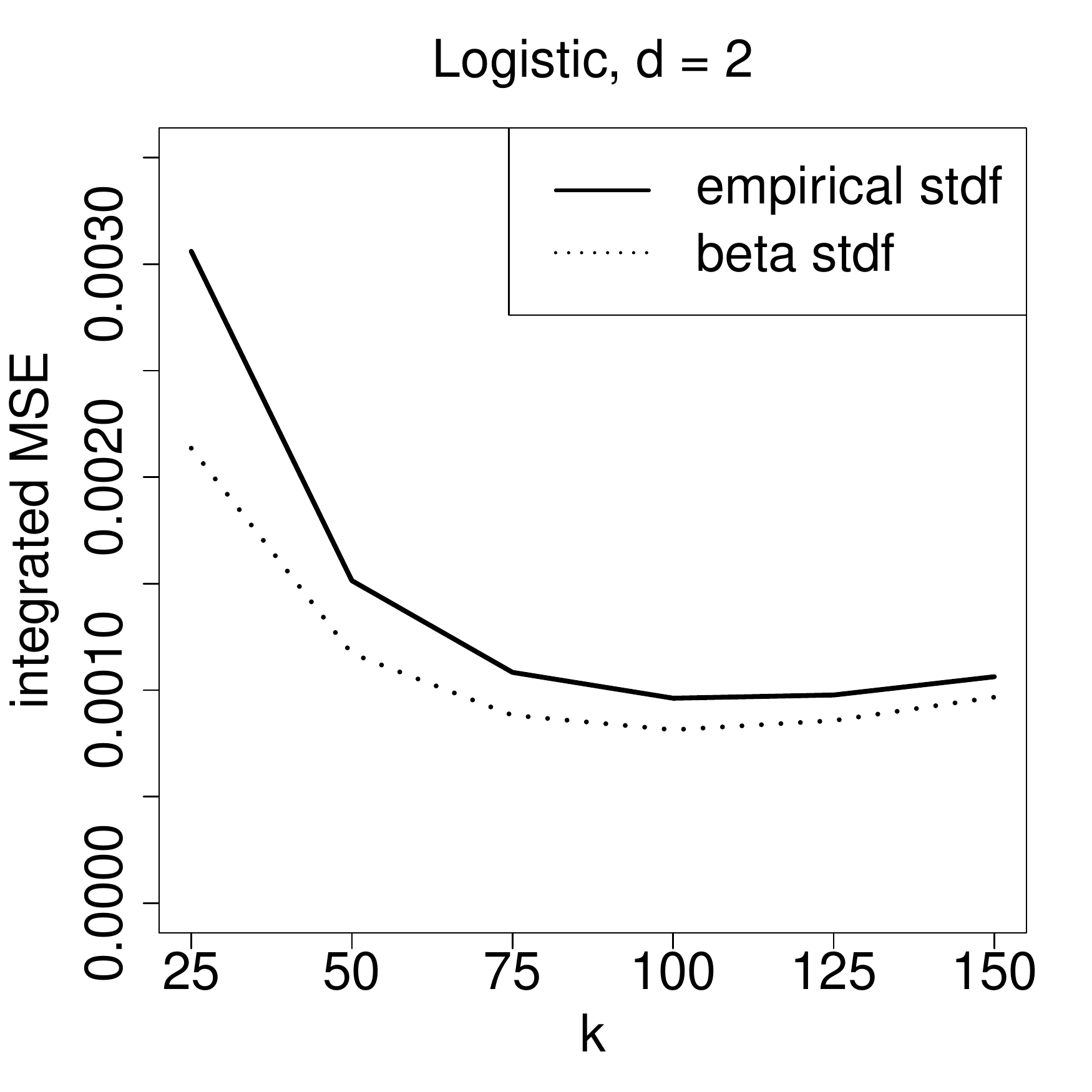}} \\
\subfloat{\includegraphics[width=0.33\textwidth]{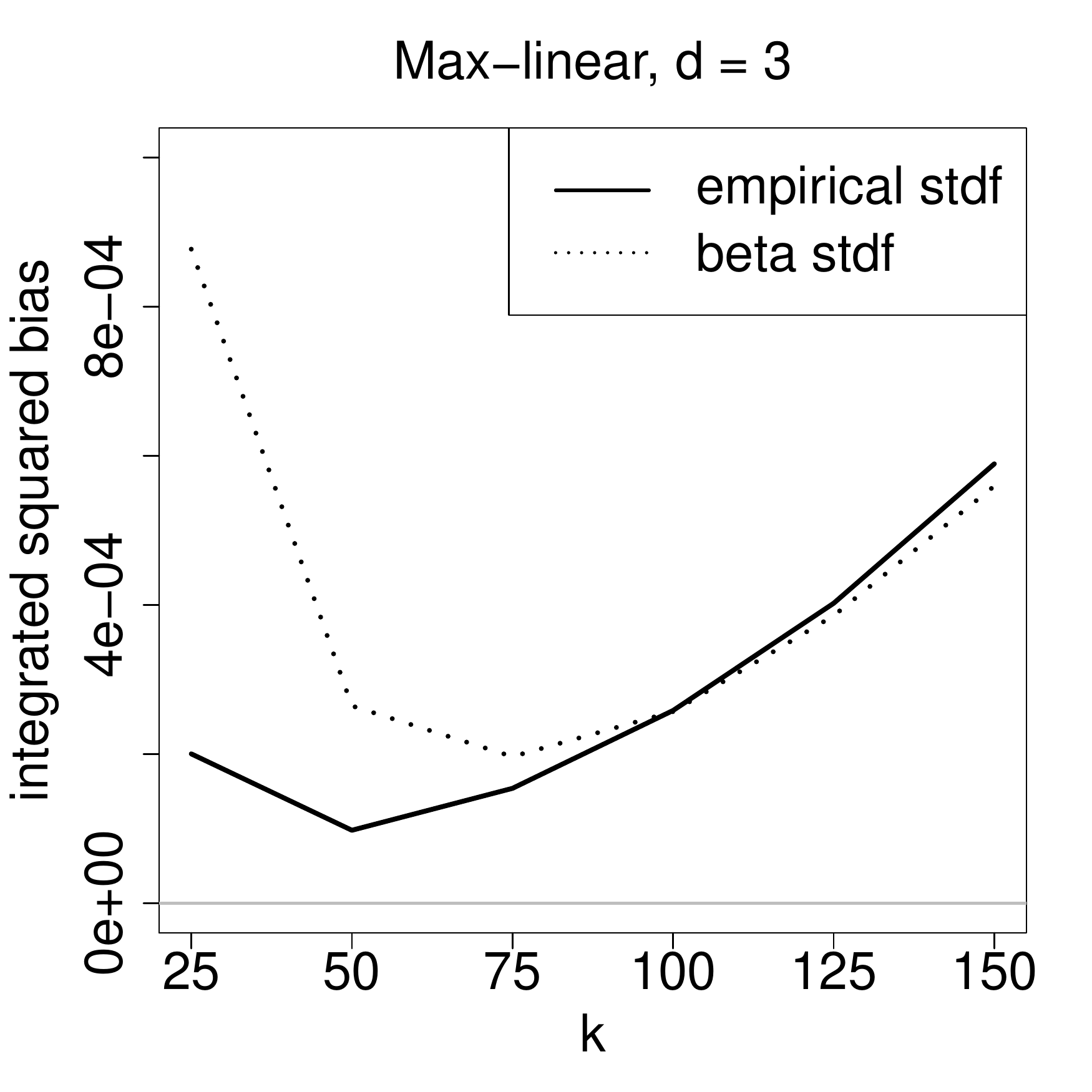}} 
\subfloat{\includegraphics[width=0.33\textwidth]{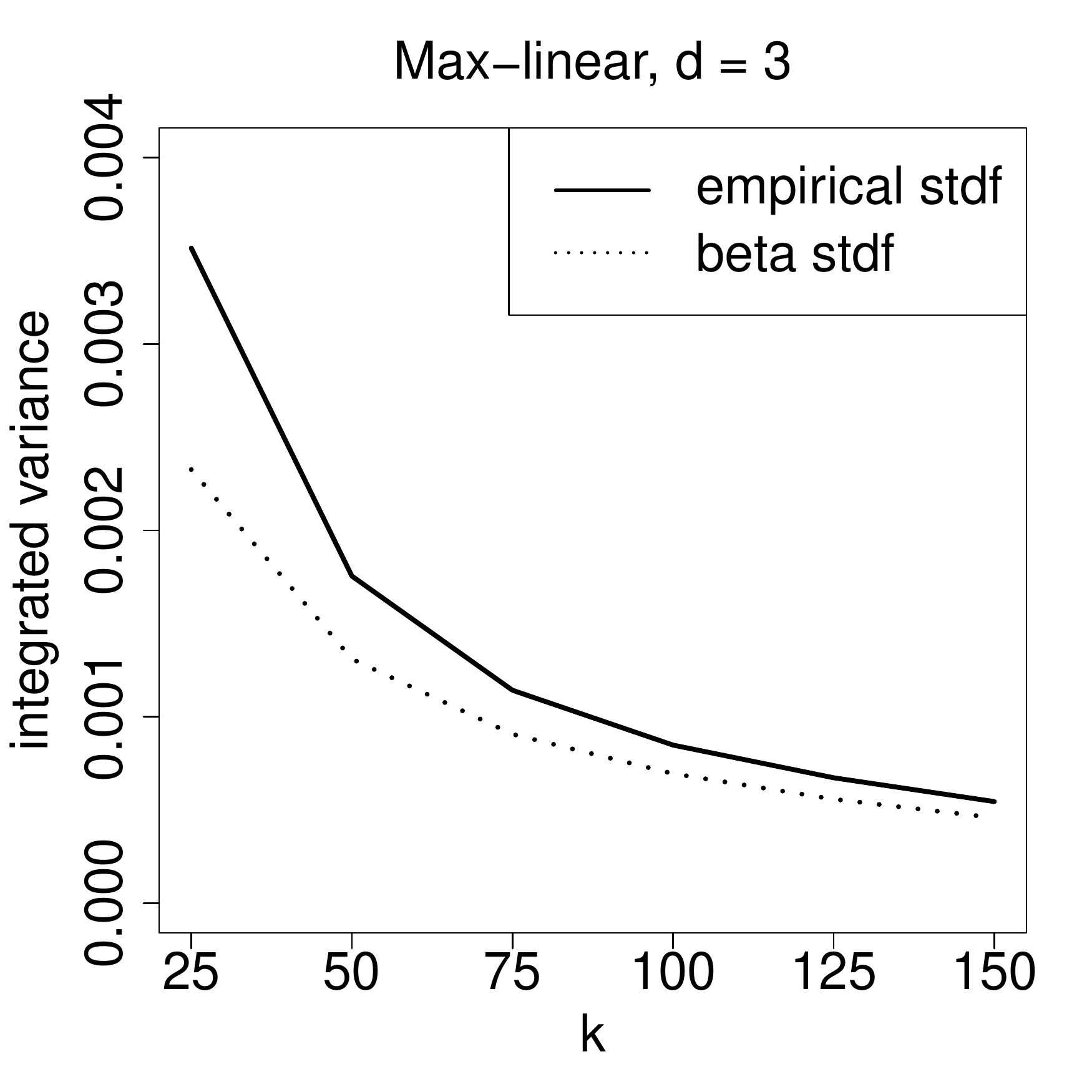}} 
\subfloat{\includegraphics[width=0.33\textwidth]{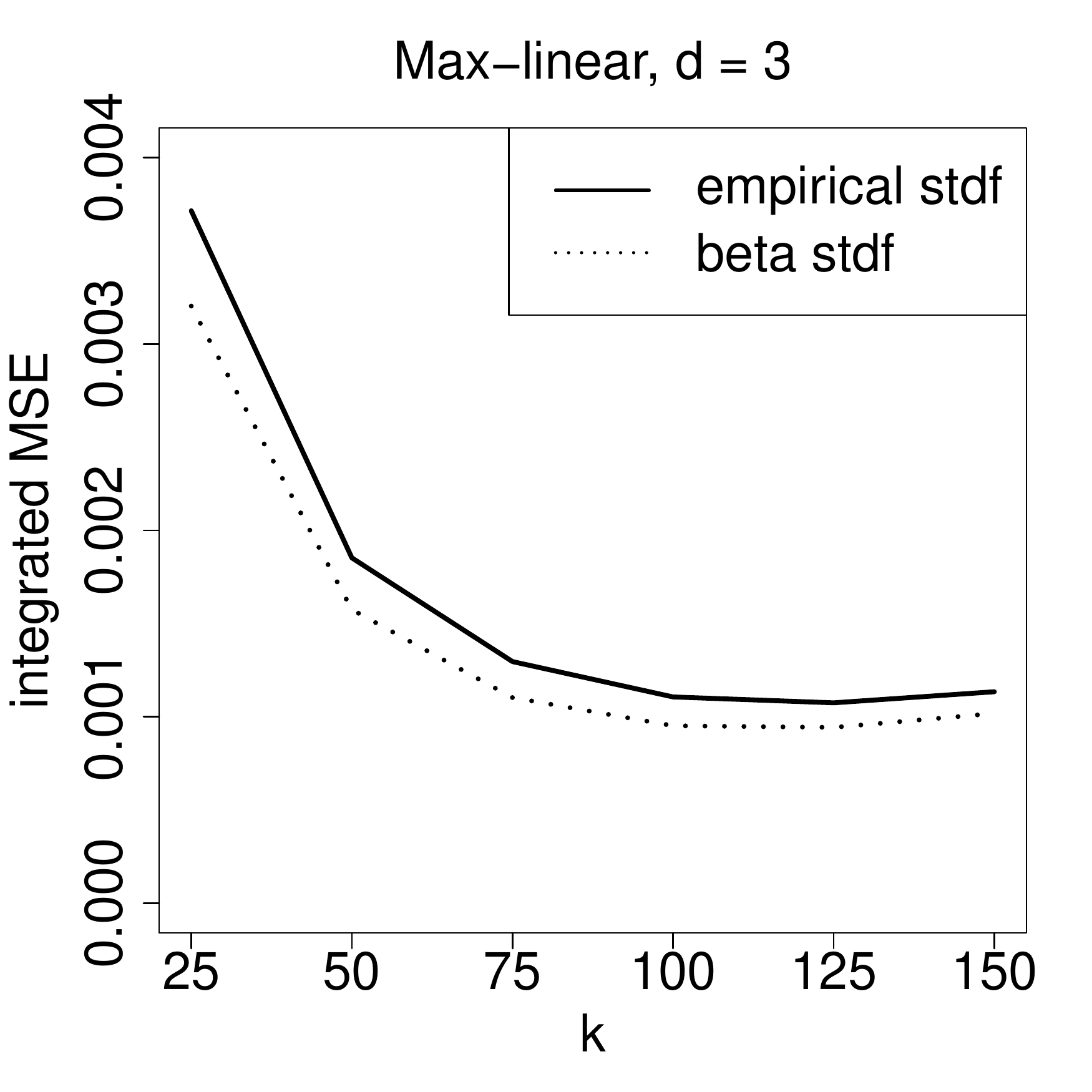}} \\
\subfloat{\includegraphics[width=0.33\textwidth]{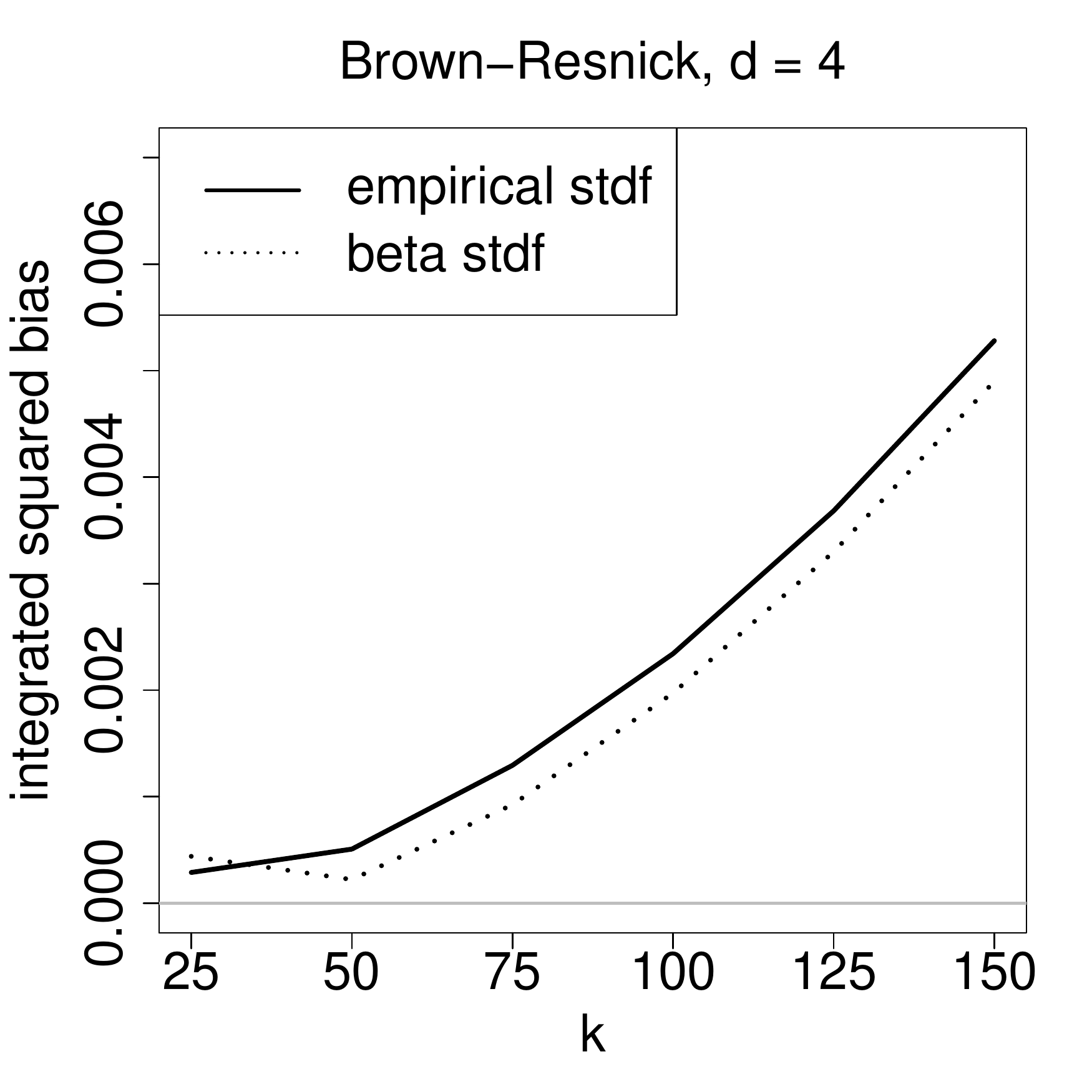}} 
\subfloat{\includegraphics[width=0.33\textwidth]{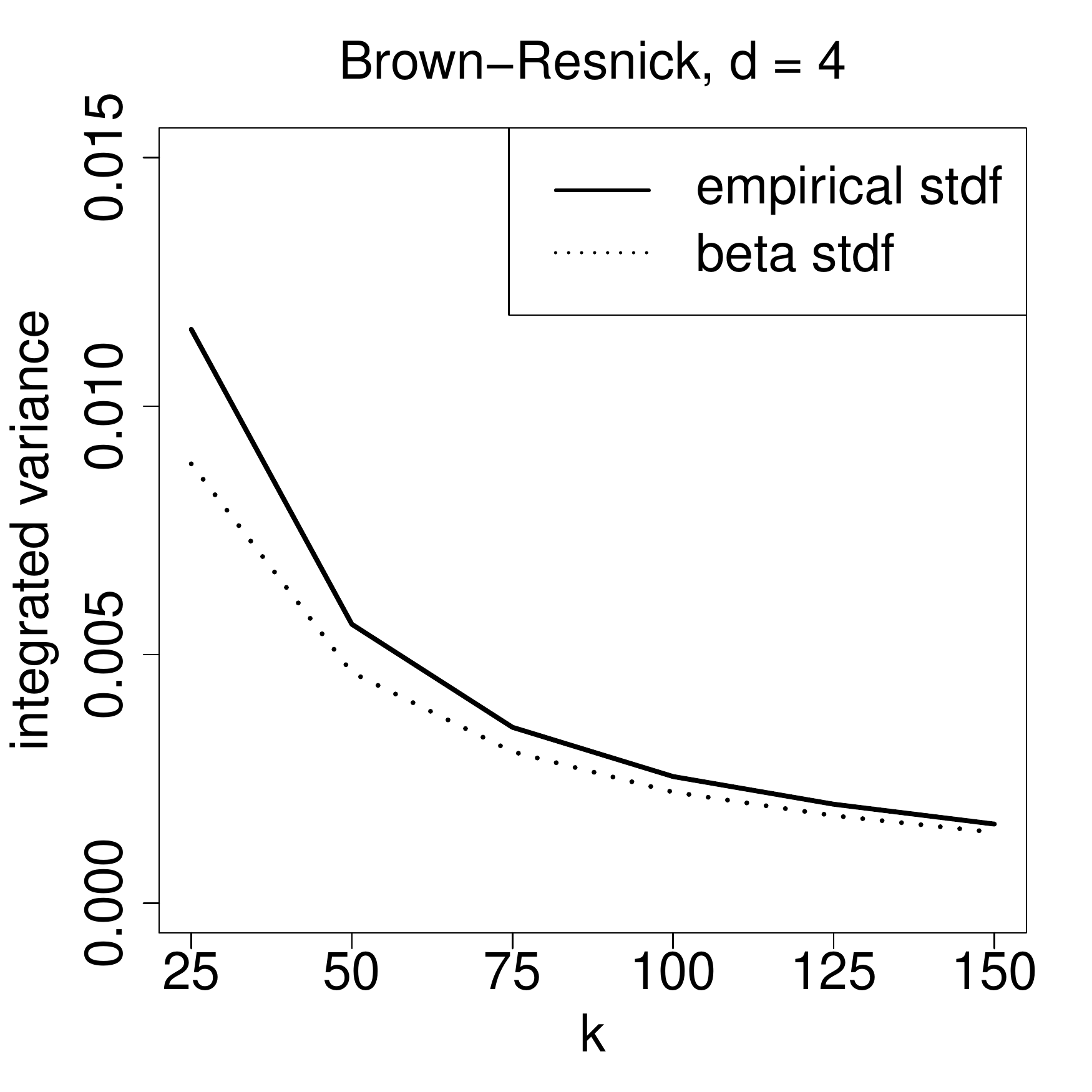}} 
\subfloat{\includegraphics[width=0.33\textwidth]{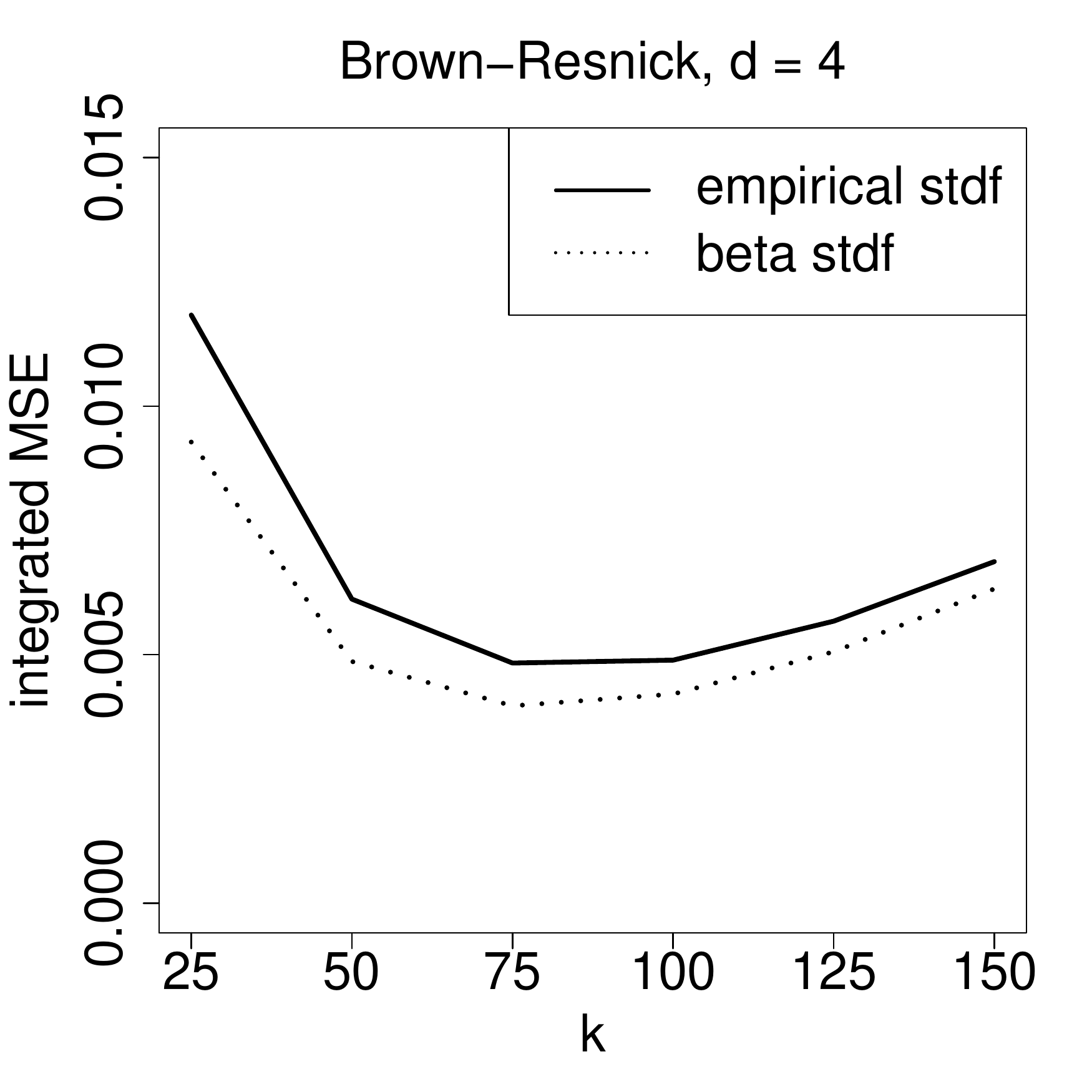}}
\caption{Integrated squared bias, variance and MSE of the empirical stdf $\ell_{n,k}$ (full line) and empirical beta stdf $\ell_{n,k}^{\beta}$ (dotted line) for the two-dimensional logistic model (upper panels), the three-dimensional max-linear model (middle panels) and the four-dimensional Brown--Resnick process (lower panels) as described in Section~\ref{sec:simul}. Plots based on 20\,000 samples of size $n = 1\,000$.} 
\label{fig:ellTot}
\end{figure}

Figure \ref{fig:ellTot} shows the integrated squared bias, integrated variance, and integrated MSE for the three models considered. Every sample had size $n = 1\,000$ and we considered $k \in \{25,50,75,100,125,150\}$. In all cases, the empirical beta stdf had a lower integrated variance and a lower integrated MSE than the empirical stdf. For the integrated bias, the results depended on the model. For the logistic and Brown--Resnick models, the empirical beta stdf had a lower integrated squared bias as well. For the max-linear model, however, the empirical beta stdf had a higher integrated squared bias, especially at smaller values of $k$. Since the marginal variance of the smoothing measure $\nu_{n,k,\vc{x}}$ in \eqref{eq:ellbeta:nu} is of the order $k^{-1}$, a smaller value of $k$ entails a larger smoothing window. For non-differentiable models such as the max-linear one, there is thus a risk of oversmoothing.

\subsection*{Semiparametric estimation}

Assume that the unknown stdf $\ell$ belongs to a parametric family $\{ \ell_\theta : \theta \in \Theta \}$, i.e., $\ell \equiv \ell_{\theta_0}$ for some unknown $\theta_0 \in \Theta$ with $\Theta \subset \mathbb{R}^p$. The aim is to estimate the unknown parameter (vector) $\theta_0$. For the logistic model~\eqref{eq:logistic}, we have $\theta \in [0, 1] = \Theta$, whereas for the Brown--Resnick process~\eqref{eq:BR}, we have $\theta = (\alpha, \rho) \in (0, 2] \times (0, \infty) = \Theta$. The model is semiparametric, since no assumptions are made on the margins except for continuity. In addition, the function $\ell$ only concerns the tail behaviour of the copula of the underlying distribution, not the whole copula.

\citet{einmahl2016nr2} proposed to estimate $\theta$ by a weighted least squares estimator,
\begin{equation}
\label{eq:wls}
  \hat{\theta}_{n,k} = \operatornamewithlimits{\arg\min}_{\theta \in \Theta}
  \sum_{r=1}^q \sum_{s=1}^q 
  \{ \ell_\theta(\vc{c}_r) - \hat{\ell}_{n,k}(\vc{c}_r) \} \,
  \Omega_{rs}(\theta) \,
  \{ \ell_\theta(\vc{c}_s) - \hat{\ell}_{n,k}(\vc{c}_s) \},
\end{equation}
where $\vc{c}_1, \ldots, \vc{c}_q \in [0, \infty)^d$ are $q$ points chosen in the domain of the stdf and where $\Omega(\theta)$ is a positive definite $q \times q$ matrix. For the latter, the identity matrix is a valid choice, while the asymptotic variance may further be reduced by a data-adaptive choice discussed in the cited article. Further, $\hat{\ell}_{n,k}$ can be any initial (nonparametric) estimator of $\ell$.

We compared the performance of the weighted least squares estimator when the initial estimator $\hat{\ell}_{n,k}$ is either the empirical stdf or the empirical beta stdf. For the bivariate logistic model~\eqref{eq:logistic}, we chose $q=4$ points $\vc{c}_m$ as $(1/2,1/2)$, $(1,1/2)$, $(1/2,1)$, and $(1,1)$, while for the four-dimensional Brown--Resnick process~\eqref{eq:BR}, we chose $q=6$ points $\vc{c}_m$ as $(1, 1, 0, 0), \ldots, (0, 0, 1, 1)$, i.e., all possible vectors with two elements equal to one and two elements equal to zero. For simplicity, the weight matrix $\Omega$ was chosen as the identity matrix.

Figure~\ref{fig:semipar} presents the root mean squared error of the weighted least squares estimator based on $500$ samples of size $n = 1\,000$ for the parameter $\theta$ of the logistic model \eqref{eq:logistic} and the parameters $\alpha$ and $\rho$ of the Brown--Resnick process \eqref{eq:BR}, with $k \in \{25, 50, \ldots, 150\}$. In all cases, the weighted least squares estimator based upon the empirical beta stdf had a lower root mean squared error than the one based upon the empirical stdf.

The use of the empirical beta stdf for inference on the parameters of (generalized) max-linear models was investigated in \citet{kiriliouk2017}. In line with our findings, the empirical beta stdf was found to yield a reduction in root mean squared error, at the expense of an increased bias due to oversmoothing.

\begin{figure}
\centering
\subfloat{\includegraphics[width=0.33\textwidth]{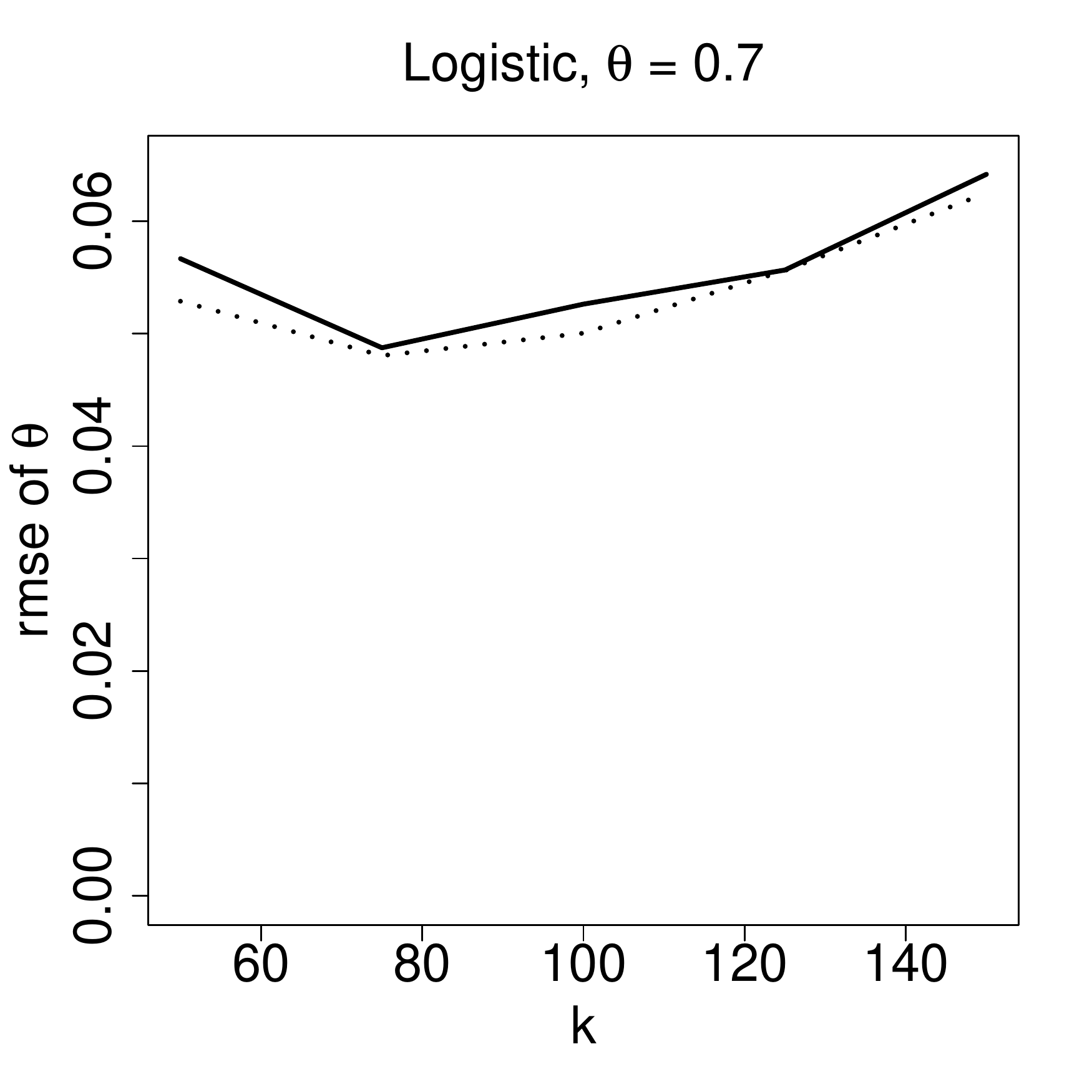}} 
\subfloat{\includegraphics[width=0.33\textwidth]{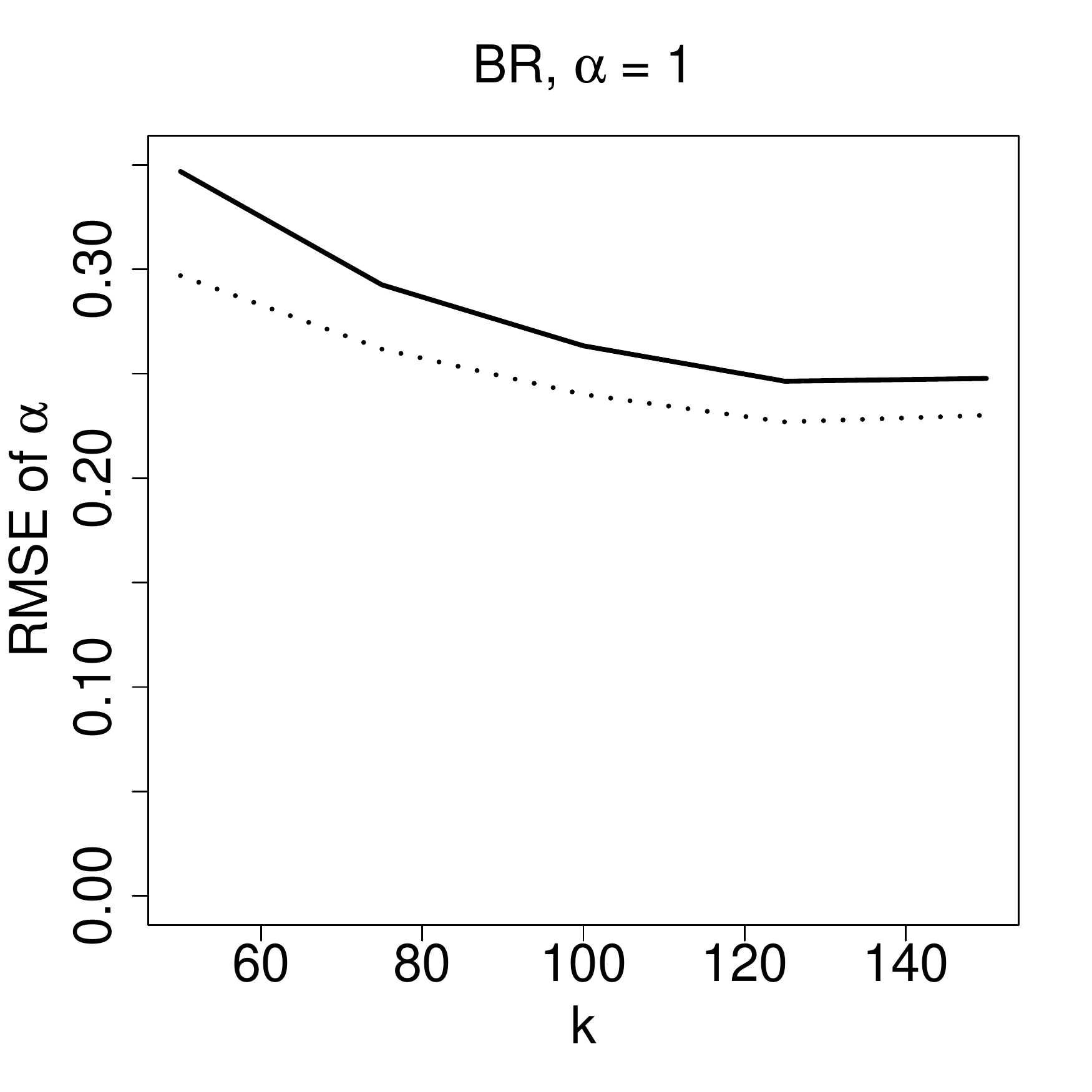}} 
\subfloat{\includegraphics[width=0.33\textwidth]{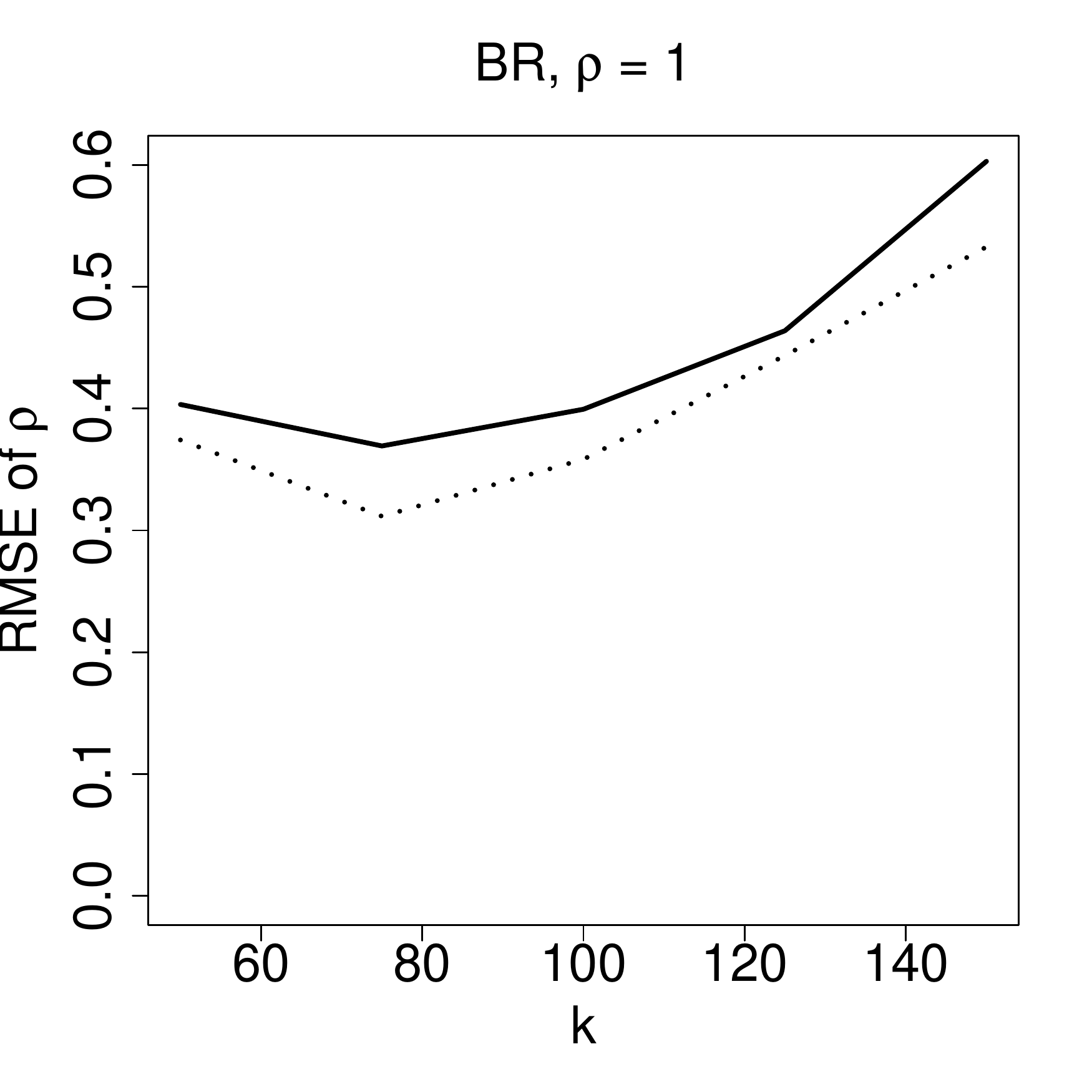}}
\caption{\label{fig:semipar} Root mean squared error of the weighted least squares estimator in \eqref{eq:wls} when the initial estimator is either the empirical stdf (solid) or the empirical beta stdf (dotted lines) for the parameter $\theta = 0.7$ (left) of the logistic model in \eqref{eq:logistic} and the parameters $\alpha = 1$ (middle) and $\rho = 1$ (right) of the Brown--Resnick process in \eqref{eq:BR}. Results based upon $500$ samples of size $n = 1\,000$ and for $k \in \{25, 50, \ldots, 150\}$.}
\end{figure}

\section{Resampling}
\label{sec:resampling}

The calculation of standard errors and of critical values of test statistics requires the distribution of the limit process $B$ in Theorem~\ref{thm:weak}. For the empirical (beta) stdf, the limiting process is Gaussian and has a covariance function that depends in a complicated way on the unknown $\ell$ and its derivatives, see e.g.~\citet[Remark~4.5 and Theorem~4.6]{einmahl2012m}. As in \citet{bucher2013nr2}, we therefore want to approximate the limiting distribution using resampling methods. To do so, we propose a method that exploits the fact that the empirical beta copula is a genuine copula. 

We first recall how to simulate a single observation from the empirical beta copula $\copb$ in \eqref{eq:empbcop} based on the sample $\vc{X}_1,\ldots,\vc{X}_n$ with marginal ranks $R_{ij,n}$. 
A single data point $\vc{U}^*$ from $\copb$ is generated as follows:
\begin{enumerate}[({A}1)]
\item Draw a random integer $I$ uniformly from $\{1,\ldots,n\}$.
\item Put $r_j = R_{Ij,n}$ for $j \in \{1, \ldots, d\}$.
\item Draw independent random variables $V_j \sim \operatorname{Beta}(r_j, n-r_j+1)$ for $j \in \{1,\ldots,d\}$.
\item Put $\vc{U}^{*} = (V_1, \ldots, V_d)$.
\end{enumerate}

In order to generate a bootstrap replicate $(R_{ij,n}^{*})_{ij}$ of the $n \times d$ matrix of ranks $(R_{ij,n})_{ij}$ and of quantities based upon it such as the empirical (beta) copula and empirical (beta) stdf, we proceed as follows:
\begin{enumerate}[({B}1)]
\item Draw an independent random sample $\vc{U}_i^* = (U_{i1}^*, \ldots, U_{id}^*)$, with $i \in \{1,\ldots,n\}$, from $\copb$, using (A1)--(A4).
\item Compute the componentwise ranks $R_{ij,n}^{*} = \sum_{t=1}^n \I \{ U_{tj}^* \le U_{ij}^* \}$ for $i \in \{1,\ldots,n\}$ and $j \in \{1,\ldots,d\}$.
\item Compute the empirical beta copula $\mathbb{C}_n^{\beta*}$ and the empirical beta stdf $\ell_{n,k}^{\beta*}$ based on the bootstrapped rank matrix $(R_{ij,n}^{*})_{ij}$.
\end{enumerate}

Then we can for instance estimate the distribution of
\[
  B_{n,k}^\beta = \sqrt{k} ( \ell_{n,k}^{\beta} - \ell )
\]
by the distribution of
\[
  B_{n,k}^{\beta*} = \sqrt{k} ( \ell_{n,k}^{\beta*} - \ell_{n,k}^{\beta} )
\] 
conditionally on the data. To compute the latter, we use a Monte Carlo approximation via the empirical distribution of a large number of independent bootstrap samples of $\ell_{n,k}^{\beta*}$ (independent conditionally on the data) generated through (B1)--(B3).


In dimension $d=2$, \citet{bucher2013nr2} proposed two multiplier bootstrap procedures for approximating the distribution of the empirical lower tail copula. The lower and upper tail copulas, $\Lambda_L$ and $\Lambda_U$, of a bivariate copula $C$ are defined as 
\begin{align*} 
  \Lambda_L(x_1, x_2) &= \lim_{t \downarrow 0} t^{-1} \, C(tx_1, tx_2), &
  \Lambda_U(x_1, x_2) &= \lim_{t \downarrow 0} t^{-1} \, \bar{C}(tx_1, tx_2),
\end{align*}
respectively, where $(x_1, x_2) \in [0, \infty)^2$ and where $\bar{C}(u, v) = u+v-1+C(1-u,1-v)$ is the survival copula associated to $C$. The upper tail copula is related to the stdf via $\Lambda_U(x_1, x_2) = x_1 + x_2 - \ell(x_1, x_2)$. Our methods can thus be applied to lower and upper tail copulas as well: changing `upper' into `lower' is merely a matter of convention, while the switch from estimators of $\Lambda_U$ to estimators of $\ell$ does not change the asymptotic distributions (up to sign), since the empirical (beta) stdf already has the correct margins [up to $\Oh_{\mathbb{P}}(k^{-1})$]. Specifically, the empirical lower tail copula and the empirical beta lower tail copula are defined as
\begin{align*}
  \Lambda_{L;n,k}(x_1, x_2) &= \tfrac{n}{k} \, \cop( \tfrac{k}{n} x_1, \tfrac{k}{n} x_2 ), &
  \Lambda_{L;n,k}^{\beta}(x_1, x_2) &= \tfrac{n}{k} \, \copb( \tfrac{k}{n} x_1, \tfrac{k}{n} x_2 ), 
\end{align*}
respectively. The resampling procedure (B1)--(B3) applies to these estimators as well and we propose to estimate the asymptotic distribution of 
\[ 
  \alpha_{n,k}^{\beta} = \sqrt{k} (\Lambda_{L;n,k}^{\beta} - \Lambda_L) 
\] 
by the one of 
\[ 
  \alpha_{n,k}^{\beta*} = \sqrt{k} (\Lambda_{L;n,k}^{\beta*} - \Lambda_{L;n,k}^{\beta})
\]
conditionally on the data, where $\Lambda_{L;n,k}^{\beta*}(x_1, x_2) = \frac{n}{k} \mathbb{C}_n^{\beta*}( \frac{k}{n} x_1, \frac{k}{n} x_2 )$ is the empirical beta lower tail copula computed from the resampled rank matrix $(R_{ij,n}^*)_{i,j}$ as in (B1)--(B3).

The two multiplier bootstrap procedures that \citet{bucher2013nr2} proposed are the direct multiplier bootstrap and the partial derivatives bootstrap. We focus on the former, as the latter was reported to perform equally well. Let $\xi_1, \ldots, \xi_n$ be independent and identically distributed, positive random variables, independent of the data, with unit mean and unit variance, and let $\bar{\xi}_n = n^{-1} \sum_{i=1}^n \xi_i$ denote their sample mean. In their Monte Carlo simulations, \citet{bucher2013nr2} took the common distribution of the multipliers as $\PP[\xi_i = 0] = \PP[\xi_i = 2] = 1/2$. \citet{bucher2013nr2} proposed to resample the empirical copula by $\mathbb{C}_n^{\xi}$ defined via
\begin{align*}
  F_n^{\xi}(x_1, x_2)
  &=
  \frac{1}{n} \sum_{i=1}^n \frac{\xi_i}{\bar{\xi}_n} \I \{ X_{i1} \le x_1, X_{i2} \le x_2 \}, \\
  F_{nj}^{\xi}(x_j)
  &=
  \frac{1}{n} \sum_{i=1}^n \frac{\xi_i}{\bar{\xi}_n} \I \{ X_{ij} \le x_j \}, \\
  \mathbb{C}_n^{\xi}(u_1, u_2)
  &=
  F_n^{\xi} \bigl( F_{n1}^{\xi-}(u_1), F_{n2}^{\xi-}(u_2) \bigr),
\end{align*}
where the generalized inverse $G^-$ of a distribution function $G$ is given by $G^-(p) = \inf \{ x \in \mathbb{R} : G(x) \ge p \}$ for $p \in (0, 1]$ and $G^-(0) = \sup \{ x \in \mathbb{R} : G(x) = 0 \}$. The direct multiplier version of the empirical lower tail copula is
\[
  \Lambda_{L;n,k}^{\xi}(x_1, x_2)
  =
  \tfrac{n}{k} \, \mathbb{C}_n^{\xi}( \tfrac{k}{n} x_1, \tfrac{k}{n} x_2 ).
\]
\citet[Theorem~3.4]{bucher2013nr2} show that the asymptotic distribution of $\alpha_{n,k}$ can be estimated consistently by the distribution of
\[
  \alpha_{n,k}^\xi = \sqrt{k} ( \Lambda_{L;n,k}^{\xi} - \Lambda_{L;n,k} )
\]
conditionally on the data.

We compared the beta resampler $\alpha_{n,k}^{\beta*}$ and the direct multiplier bootstrap $\alpha_{n,k}^{\xi}$ in a numerical experiment with the same settings as the ones in \citet[Section~3.2]{bucher2013nr2}. Specifically, data were generated from the bivariate Clayton copula $C_\theta(u_1, u_2) = (u_1^{-\theta} + u_2^{-\theta} - 1)^{-1/\theta}$ with parameter $\theta = 0.5$. Its lower tail copula is $\Lambda_L(x_1, x_2) = (x_1^{-\theta} + x_2^{-\theta})^{-1/\theta}$, which is the negative logistic model of \cite{joe1990}.

We generated $1\,000$ samples of size $n = 1\,000$, computed estimators at $k = 50$, and considered the resampling methods with $500$ bootstrap replications per sample. We compared the two resampling methods on the basis of the accuracy with which they estimate the true covariance matrix of the trivariate Gaussian random vector that arises by evaluating the limit process, $\alpha$, of $\alpha_{n,k}$ at the points $(x_{1m}, x_{2m}) = (\cos(m\pi/8), \sin(m\pi/8))$ for $m \in \{1, 2, 3\}$. For each data sample, the true $3 \times 3$ covariance matrix was estimated via the sample covariance matrix over the $500$ bootstrap replications. For each of the two methods, we thus obtain $1\,000$ estimates of the true covariance matrix.

Table~\ref{tab:boot} shows the following information:
\begin{itemize}
\item columns 1--3: the true covariance matrix of the limit process $\alpha$ evaluated at $(x_{1m}, x_{2m})_{m=1,2,3}$ -- copied from \citet[Table~1]{bucher2013nr2};
\item columns 4--6: the average of the $1\,000$ estimates of the covariance matrix via the beta resampler $\alpha_{n,k}^{\beta*}$ (lines 1--3) and the mean squared error (MSE) of these estimates (lines 4--6);
\item columns 7--9: idem as in columns 4--6, now for the direct multiplier bootstrap $\alpha_{n,k}^{\xi}$ -- copied from \citet[Table~2]{bucher2013nr2}.
\end{itemize}
We see that the beta resampler is more accurate than the direct multiplier bootstrap.

\begin{table}
\centering
\begin{tabular}{cccccccccccccc}
\toprule
\multicolumn{1}{l}{} & \multicolumn{3}{c}{True} 
& \multicolumn{2}{l}{} & \multicolumn{3}{c}{$\alpha_{n,k}^\beta$} 
& \multicolumn{2}{l}{} & \multicolumn{3}{c}{$\alpha_{n,k}^{\xi}$} \\
\cmidrule(r){2-4}
\cmidrule(r){7-9}
\cmidrule(r){12-14}
\text {} & $\tfrac{\pi}{8}$  & $\tfrac{2 \pi}{8}$  & $\tfrac{3 \pi}{8}$  
& & & $\tfrac{\pi}{8}$  & $\tfrac{2 \pi}{8}$  & $\tfrac{3 \pi}{8}$  
& & & $\tfrac{\pi}{8}$  & $\tfrac{2 \pi}{8}$  & $\tfrac{3 \pi}{8}$  \\
\midrule
$\tfrac{\pi}{8}$ &  0.087 & 0.075 & 0.052 & & 
$\tfrac{\pi}{8}$ &  0.087 & 0.084 & 0.049 & & 
$\tfrac{\pi}{8}$ &  0.100 & 0.071 & 0.045 \\[1ex]
$\tfrac{2 \pi}{8}$ &  & 0.116 & 0.075 & & 
$\tfrac{2 \pi}{8}$ &  & 0.106 & 0.069 & & 
$\tfrac{2 \pi}{8}$ &  & 0.136 & 0.071  \\[1ex]
$\tfrac{3 \pi}{8}$ & & & 0.087 & & 
$\tfrac{3 \pi}{8}$ & & & 0.077 & & 
$\tfrac{3 \pi}{8}$ & & & 0.099  \\
\midrule
$\tfrac{\pi}{8}$ &  & & & & 
$\tfrac{\pi}{8}$ &  1.91 & 3.11 & 1.71 & & 
$\tfrac{\pi}{8}$ &  3.86 & 3.49 & 2.72 \\[1ex]
$\tfrac{2 \pi}{8}$ &  & & & & 
$\tfrac{2 \pi}{8}$ &  & 4.28 & 2.34 & & 
$\tfrac{2 \pi}{8}$ &  & 8.89 & 3.25 \\[1ex]
$\tfrac{3 \pi}{8}$  & & & & & 
$\tfrac{3 \pi}{8}$  & & & 2.34 & & 
$\tfrac{3 \pi}{8}$  & & & 3.77 \\
\bottomrule 
\end{tabular}
\caption{\label{tab:boot} \small Columns 1--3: true covariance matrix of the limit process $\alpha$ evaluated at three points on the unit circle when the underlying copula is Clayton with parameter $\theta = 0.5$. Columns 4--6: average of $1\,000$ estimates of the covariance matrix (rows 1--3) and their $\operatorname{MSE} \times 10^4$ (rows 4--6) based on samples of size $n = 1\,000$ calculated at $k = 50$ and computed using the beta resampler $\alpha_{n,k}^{\beta*}$, each estimate being given by the sample covariance matrix obtained from $500$ bootstrap replications. Columns 7--9: idem for the direct multiplier bootstrap $\alpha_{n,k}^{\xi}$. Columns 1--3 and 7--9 have been copied from \citet[Tables~1 and~2]{bucher2013nr2}.}  
\end{table}

\appendix

\section{Proofs of Propositions~\ref{prop:stoch} and~\ref{prop:bias}}
\label{sec:proofs}

\begin{proof}[Proof of Proposition~\ref{prop:stoch}]
Fix $\eps \in (0, \delta]$. Since $\nu_{n,k,\vc{x}}$ is a probability measure, we can bring the term $B_{n,k}(\vc{x})$ inside the integral. Split the integral according to the two cases $\lvert \vc{y} - \vc{x} \rvert_\infty \le \eps$ or $\lvert \vc{y} - \vc{x} \rvert_\infty > \eps$, where $\lvert \vc{z} \vert_\infty = \max ( \lvert z_1 \rvert, \ldots, \lvert z_d \rvert )$ for $\vc{z} \in \mathbb{R}^d$. For $\vc{x} \in [0, 1]^d$, the absolute value in \eqref{eq:BnknuB} is bounded by
\begin{multline}
\label{eq:BnknuB:split}
  \sup \left\{
    \lvert B_{n,k}(\vc{y}) - B_{n,k}(\vc{x}) \rvert 
    \; : \;
    \vc{y} \in [0, n/k]^d, \, \lvert \vc{y} - \vc{x} \rvert_\infty \le \eps 
  \right\} \\
  +
  2 \sup_{\vc{y} \in [0, n/k]^d} \lvert B_{n,k}( \vc{y} ) \rvert
  \cdot
  \nu_{n,k,\vc{x}} \left( 
    \{ \vc{y} \in [0, n/k]^d : \lvert \vc{y} - \vc{x} \rvert_\infty > \eps \} 
  \right).
\end{multline}

In the first term in \eqref{eq:BnknuB:split}, we have $\vc{x} \in [0, 1]^d$, $\vc{y} \in [0, n/k]^d$, and $\lvert \vc{y} - \vc{x} \rvert_\infty \le \eps \le \delta$, whence $\vc{y} \in [0, 1+\delta]^d$. The supremum is thus bounded by the maximal increment of $B_{n,k}$ on $[0, 1+\delta]^d$ between points at a distance at most $\eps$ apart, i.e.,
\[
  \omega( B_{n,k}, \eps )
  =
  \sup \{ 
    \lvert B_{n,k}(\vc{y}_1) - B_{n,k}(\vc{y}_2) \rvert :
    \vc{y}_1, \vc{y}_2 \in [0, 1+\delta]^d, \,
    \lvert \vc{y}_1 - \vc{y}_2 \rvert_\infty \le \eps
  \}.  
\]
By Condition~\ref{cond:B}, we can find for every $\eta > 0$ a sufficiently small $\eps > 0$ such that
\[
  \limsup_{n \to \infty}
  \PP \left[ \omega( B_{n,k}, \eps ) > \eta \right] \le \eta.
\]
The first term in \eqref{eq:BnknuB:split} can thus be made arbitrarily small with arbitrarily large probability, uniformly in $\vc{x} \in [0, 1]^d$ and for sufficiently large $n$.

For the second term in \eqref{eq:BnknuB:split}, note first that
\[
  \sup_{\vc{y} \in [0, n/k]^d} \lvert B_{n,k}( \vc{y} ) \rvert
  = \Oh_{\PP}( n / \sqrt{k} ), \qquad n \to \infty.
\]
Indeed, since $\ell$ is a stdf, we have $0 \le \ell(\vc{y}) \le y_1+\cdots+y_d \le dn/k$ for $\vc{y} \in [0, n/k]^d$; for the pilot estimator $\hat{\ell}_{n,k}$, use Condition~\ref{cond:growth}.

If $S$ is a $\operatorname{Bin}(n, u)$ random variable, Bennett's inequality \citep[Proposition~A.6.2]{van1996weak} states that
\[
  \PP[ \sqrt{n} \lvert S/n - u \rvert \ge \lambda ]
  \le
  2 \exp \left\{ - nu \, h \left( 1 + \frac{\lambda}{\sqrt{n} u} \right) \right\},
  \qquad
  \lambda > 0,
\]
where $h(1+\eta) = \int_0^\eta \log(1+t) \diff t$ for $\eta \ge 0$. Note that $h(1+\eta) > \frac{1}{3} \eta^2$ for $\eta \in [0, 1]$. It follows that
\begin{align*}
  \nu_{n,k,\vc{x}} \left( 
    \{ \vc{y} \in [0, n/k]^d : \lvert \vc{y} - \vc{x} \rvert_\infty > \eps \} 
  \right)
  &\le
  \sum_{j=1}^d
  \PP\left[ 
    \left\lvert 
      \operatorname{Bin}(n, \tfrac{k}{n} x_j) / k - x_j 
    \right\rvert > \eps 
  \right] \\
  &\le
  \sum_{j=1}^d
  \PP\left[ 
     \sqrt{n}
      \left\lvert 
	\operatorname{Bin}(n, \tfrac{k}{n} x_j) / n 
	- 
	\tfrac{k}{n} x_j 
      \right\rvert 
    > 
    k\eps/\sqrt{n} 
  \right] \\
  &\le
  \sum_{j=1}^d
  2 \exp \left\{ 
    - kx_j \, h \left( 1 + \frac{k\eps/\sqrt{n}}{\sqrt{n} kx_j/n} \right) 
  \right\} \\
  &=
  \sum_{j=1}^d
  2 \exp \left\{ 
    - kx_j \, h ( 1 + \eps/x_j )
  \right\}.
\end{align*}
As $\partial \{ x \, h(1+\eps/x) \} / \partial x < 0$ for $0 < x < 1$, we have $\inf_{x \in [0, 1]} \{ x \, h(1+\eps/x) \} = h(1+\eps)$. We conclude that
\[
  \nu_{n,k,\vc{x}} \left( 
    \{ \vc{y} \in [0, n/k]^d : \lvert \vc{y} - \vc{x} \rvert_\infty > \eps \} 
  \right)
  \le
  2d \exp \{ - k \, h(1+\eps) \}
  \le
  2d \exp \left( - \tfrac{1}{3} k \eps^2 \right).
\]
Together, the supremum over $\vc{x} \in [0, 1]^d$ of the second term in \eqref{eq:BnknuB:split} is of the order
\[
  \Oh_{\PP} \left( \frac{n}{\sqrt{k}} \exp \left( - \tfrac{1}{3} k \eps^2 \right) \right),
  \qquad n \to \infty.
\]
It therefore converges to zero in probability since $\log(n) = \oh(k)$ by assumption.
\end{proof}

\begin{proof}[Proof of Proposition~\ref{prop:bias}]
Fix $\vc{x} \in [0, M]^d$. For $j \in \{1, \ldots, d\}$ such that $x_j = 0$, the binomial distribution $\operatorname{Bin}(n, (k/n)x_j)$ is concentrated on $0$. As a consequence, the integral over $\vc{y} \in [0, n/k]^d$ with respect to $\nu_{n,k,\vc{x}}$ can be restricted to the set of those $\vc{y} \in [0, n/k]^d$ such that $y_j = 0$ for all $j \in \{1, \ldots, d\}$ for which $x_j = 0$. Call this set $\mathbb{D}(n,k,\vc{x})$.

For $\vc{y} \in \mathbb{D}(n,k,\vc{x})$, the function $[0, 1] \to \RR : t \mapsto f(\vc{x} + t(\vc{y}-\vc{x}))$ is continuous on $[0, 1]$ and is continuously differentiable on $(0, 1)$; indeed, if $x_j = 0$, then the $j$th component of $\vc{x} + t(\vc{y}-\vc{x})$ vanishes and thus does not depend on $t \in [0, 1]$, while if $x_j > 0$, then that component is (strictly) positive for all $t \in [0, 1)$, so that, by assumption, $\dot{f}_j(\vc{x} + t(\vc{y}-\vc{x}))$ exists and is continuous in $t \in [0, 1)$. Writing $J(\vc{x}) = \{j = 1, \ldots, d: x_j > 0\}$, we find, by the fundamental theorem of calculus,
\[
  f( \vc{y} ) - f( \vc{x} )
  =
  \sum_{j \in J(\vc{x})}
  (y_j - x_j) \int_0^1 \dot{f}_j \bigl( \vc{x} + t(\vc{y} - \vc{x}) \bigr) \diff t,
  \qquad
  \vc{y} \in \mathbb{D}(n,k,\vc{x}).
\]
We obtain
\begin{align*}
  \Delta_{n,k}(\vc{x})
  &:=
  \int_{\mathbb{D}(n,k,\vc{x})} \sqrt{k} \{ f(\vc{y}) - f(\vc{x}) \} \diff \nu_{n,k,\vc{x}}( \vc{y} ) 
  \\
  &=
  \sum_{j \in J(\vc{x})}
  \int_{\vc{y} \in \mathbb{D}(n,k,\vc{x})}
    \sqrt{k} (y_j - x_j) 
    \int_{0}^1 
      \dot{f}_j \bigl( \vc{x} + t(\vc{y} - \vc{x}) \bigr) 
    \diff t
  \diff \nu_{n,k,\vc{x}}(\vc{y}) \\
  &=
  \sum_{j \in J(\vc{x})}
  \int_{\vc{y} \in \mathbb{D}(n,k,\vc{x})}
     \sqrt{k} (y_j - x_j)
    \int_{0}^1 
      \bigl\{ 
	\dot{f}_j \bigl( \vc{x} + t(\vc{y} - \vc{x}) \bigr) - \dot{f}_j(\vc{x}) 
      \bigr\}
    \diff t
  \diff \nu_{n,k,\vc{x}}(\vc{y}),
\end{align*}
where the last step is justified via $\int y_j \diff \nu_{n,k,\vc{x}}(\vc{y}) = \EE[ \operatorname{Bin}(n, (k/n)x_j)/k ] = x_j$. Taking absolute values, we find, for $\vc{x} \in [0, M]^d$,
\[
  \left\lvert
    \Delta_{n,k}(\vc{x})
  \right\rvert
  \le
  \sum_{j \in J(\vc{x})}
  I_{n,k}(\vc{x}, j)
\]
where
\[
  I_{n,k}(\vc{x}, j)
  =
  \int_{\vc{y} \in \mathbb{D}(n,k,\vc{x})}
    \sqrt{k} \lvert y_j - x_j \rvert
    \int_{0}^1
      \left\lvert
	\dot{f}_j \left( \vc{x} + t(\vc{y} - \vc{x}) \right)
	-
	\dot{f}_j( \vc{x} )
      \right\rvert
    \diff t
  \diff \nu_{n,k,\vc{x}}(\vc{y}).
\]
We will find an upper bound for $I_{n,k}(\vc{x},j)$. 

Let $K > 0$ be such that $\lvert \dot{f}_i \rvert \le K$ for all $i \in \{1, \ldots, d\}$. Choose $\delta \in (0, M]$ and $\eps \in (0, \delta/2]$. In $I_{n,k}(\vc{x},j)$, split the integral over $\vc{y} \in \mathbb{D}(n,k,\vc{x})$ into two pieces, depending on whether $\lvert \vc{y} - \vc{x} \rvert \le \eps$ or $\lvert \vc{y} - \vc{x} \rvert > \eps$, where $\lvert \vc{z} \rvert = (z_1^2 + \cdots + z_d^2)^{1/2}$ denotes the Euclidean norm of $\vc{z} \in \RR^d$.

In $I_{n,k}(\vc{x},j)$, the integral over $\vc{y} \in \mathbb{D}(n,k,\vc{x})$ for which $\lvert \vc{y} - \vc{x} \rvert > \eps$ is bounded by
\begin{align*}
  2K \sqrt{k}
  \int_{\mathbb{D}(n,k,\vc{x})} 
    \lvert y_j - x_j \rvert 
    \I\{ \lvert \vc{y} - \vc{x} \rvert > \eps \} 
  \diff \nu_{n,k,\vc{x}}(\vc{y}) 
  &\le
  \frac{2K\sqrt{k}}{\eps} 
  \int_{\mathbb{D}(n,k,\vc{x})}
  \lvert \vc{y} - \vc{x} \rvert^2
  \diff \nu_{n,k,\vc{x}}(\vc{y}) \\
  &=
  \frac{2K\sqrt{k}}{\eps} 
  \sum_{i=1}^d
  \frac{1}{k^2} \cdot n \cdot \tfrac{k}{n} x_i \cdot (1 - \tfrac{k}{n} x_i) \\
  &\le
  \frac{2KMd}{\eps \sqrt{k}}.
\end{align*}

To analyse the integral in $I_{n,k}(\vc{x}, j)$ over those $\vc{y} \in \mathbb{D}(n,k,\vc{x})$ for which $\lvert \vc{y} - \vc{x} \rvert \le \eps$, we need to distinguish between two cases: $x_j < \delta$ and $x_j \ge \delta$. In case $x_j < \delta$, the integral is simply bounded by
\begin{align*}
  2K 
  \int_{\mathbb{D}(n,k,\vc{x})} 
    \sqrt{k} \lvert y_j - x_j \rvert 
  \diff \nu_{n,k,\vc{x}}(\vc{y}) 
  &\le
  2K \sqrt{k} 
  \sqrt{ 
    \tfrac{1}{k^2} 
    \cdot n 
    \cdot \tfrac{k}{n}x_j 
    \cdot (1 - \tfrac{k}{n}x_j)
  } \\
  &\le
  2K \sqrt{x_j}
  <
  2K \sqrt{\delta}.
\end{align*}
In case $x_j \ge \delta$, the inequality $\lvert \vc{y} - \vc{x} \rvert \le \eps \le \delta/2$ and the fact that $\vc{x} \in [0, M]^d$ and $\vc{y} \in [0, \infty)^d$ imply that $\vc{y}$ belongs to the set
\[
  \mathbb{B}_j(M, \delta)
  =
  \{ \vc{z} \in [0, M+\delta/2]^2 : z_j > \delta/2 \}.
\]
Let
\[
  \omega_j(M, \delta, \eps)
  =
  \sup \{
    \lvert 
      \dot{f}_j( \vc{z}_1 ) 
      - 
      \dot{f}_j( \vc{z}_2 ) 
    \rvert :
    \vc{z}_1, \vc{z}_2 \in \mathbb{B}_j(M, \delta), \,
    \lvert \vc{z}_1 - \vc{z}_2 \rvert \le \eps
  \}.
\]
The integral in $I_{n,k}(\vc{x}, j)$ over $\vc{y} \in \mathbb{D}(n,k,\vc{x})$ such that $\lvert \vc{y} - \vc{x} \rvert \le \eps$ is bounded by
\[
  \omega_j(M, \delta, \eps)
  \int_{\mathbb{D}(n,k,\vc{x})}
    \sqrt{k} \lvert y_j - x_j \rvert
  \diff \nu_{n,k,\vc{x}}(\vc{y})
  \le
  \omega_j(M, \delta, \eps)
  \sqrt{x_j}
  \le
  \omega_j(M, \delta, \eps)
  \sqrt{M}
\]
using the Cauchy--Schwarz inequality and the first two moments of the binomial distribution.

Assembling all the pieces, we obtain
\[
  \sup_{\vc{x} \in [0, M]^d}
  \left\lvert
    \Delta_{n,k}(\vc{x})
  \right\rvert \\
  \le
  \frac{2d^2KM}{\eps \sqrt{k}}
  +
  2dK \sqrt{\delta}
  +
  \sqrt{M}
  \sum_{j=1}^d \omega_j(M, \delta, \eps).
\]
As a consequence, for every $\delta \in (0, M]$ and every $\eps \in (0, \delta/2]$, we have
\[
  \limsup_{n \to \infty}
  \sup_{\vc{x} \in [0, M]^d}
  \left\lvert
    \Delta_{n,k}(\vc{x})
  \right\rvert 
  \le
  2dK \sqrt{\delta}
  +
  \sqrt{M}
  \sum_{j=1}^d \omega_j(M, \delta, \eps).
\]
The function $\dot{f}_j$ is continuous and thus uniformly continuous on the compact set $\mathbb{B}_j(M, \delta)$. As consequence, $\inf_{\eps > 0} \omega_j(M, \delta, \eps) = 0$. The limit superior in the previous display is thus bounded by $2dK \sqrt{\delta}$, for all $\delta \in (0, M]$, and must therefore be equal to zero. 
\end{proof}

\section*{Acknowledgments}

A.~Kiriliouk gratefully acknowledges support from the Fonds de la Recherche Scientifique (FNRS).

J.~Segers gratefully acknowledges funding by contract ``Projet d'Act\-ions de Re\-cher\-che Concert\'ees'' No.\ 12/17-045 of the ``Communaut\'e fran\c{c}aise de Belgique'' and by IAP research network Grant P7/06 of the Belgian government (Belgian Science Policy).

L.~Tafakori would like to thank the Australian Research Council for supporting this work through Laureate
Fellowship FL130100039.

\bibliographystyle{chicago} 
\bibliography{references} 

\end{document}